\documentclass[12pt]{article}
\usepackage[dvips]{graphics}
\usepackage{booktabs}
\usepackage{setspace}
\usepackage[final]{graphicx}
\usepackage[dvipsnames]{color}
\usepackage[round]{natbib}
\usepackage{amsmath}
\usepackage[colorlinks, citecolor=black]{hyperref}
\usepackage{mathrsfs}
\usepackage{bm, xcolor}
\usepackage{mathrsfs}
\usepackage{verbatim}
\usepackage{threeparttable}
\usepackage{color}
\usepackage[utf8]{inputenc}
\usepackage{float}
\usepackage{listings}
\usepackage{amssymb}
\usepackage{amsthm}
\usepackage{multirow}
\usepackage{caption}
\usepackage{filecontents}
\usepackage{geometry}
\usepackage{enumitem}
\usepackage{algorithm}
\usepackage{algpseudocode}
\usepackage{subfigure}
\renewcommand{\baselinestretch}{1.6}

\makeatletter
\def\singlespace{\def\baselinestretch{1}\@normalsize}

\newtheorem{lemma}{Lemma}
\newtheorem{theorem}{Theorem}
\newtheorem{remark}{Remark}
\newtheorem{corollary}{Corollary}

\renewcommand{\hat}{\widehat}

\makeatother

\newcommand{\be}{\begin{equation}}
\newcommand{\ee}{\end{equation}}
\newcommand{\beaa}{\begin{eqnarray*}}
\newcommand{\eeaa}{\end{eqnarray*}}
\newcommand{\bea}{\begin{eqnarray}}
\newcommand{\eea}{\end{eqnarray}}

\newcommand{\bY}{\mathbf{Y}}
\newcommand{\bX}{\mathbf{X}}
\newcommand{\bU}{\mathbf{U}}

\newcommand{\bZ}{\mathbf{Z}}

\newcommand{\bV}{\mathbf{V}}
\newcommand{\bW}{\mathbf{W}}
\newcommand{\bepsilon}{\bm{\varepsilon}}
\newcommand{\bfeta}{\bm{\eta}}
\newcommand{\mbE}{\mathbb{E}}

\newcommand{\btheta}{\bm{\theta}}
\newcommand{\bgamma}{\bm{\gamma}}
\newcommand{\bSigma}{\bm{\Sigma}}

\begin{document}

\title{
Double-Estimation-Friendly Inference for High-Dimensional Measurement Error Models with Non-Sparse Adaptability}

\author{Shijie Cui$^a$, Xu Guo$^b$, Songshan Yang$^c$ and Zhe Zhang$^d$\\
 $^a$Pennsylvania State University, $^b$Beijing Normal University\\
$^c$Renmin University of China,
$^d$University of North Carolina at Chapel Hill}

\date{}
\maketitle
\renewcommand{\baselinestretch}{1.1}
\baselineskip=18pt
\begin{abstract}
{In this paper, we introduce an innovative testing procedure for assessing individual hypotheses in high-dimensional linear regression models with measurement errors. This method remains robust even when either the $X$-model or $Y$-model is misspecified. We develop a double robust score function that maintains a zero expectation if one of the models is incorrect, and we construct a corresponding score test. We  show the asymptotic normality of our approach in a low-dimensional setting, and then extend it to the high-dimensional models. Our analysis of high-dimensional settings explores scenarios both with and without the sparsity condition, establishing asymptotic normality and non-trivial power performance under local alternatives. Simulation studies and real data analysis demonstrate the effectiveness of the proposed method.
}

\end{abstract}

\par \vspace{9pt}
\noindent%
\vfill

\section{Introduction}
Measurement errors are common in the real world. For example, the measurements of indexes of patients in clinical trials that are difficult to measure or that rely on self-reporting are often inaccurate. The validity of any medical study can be potentially affected by measurement error. Moreover, measurement error can also occur in high dimensional modeling scenarios. For instance, consider a case where we have data on a patient's self-reported weekly alcohol consumption along with the expression levels of tens of genes from the same patient. The self-reported weekly alcohol consumption of the patient can have possible measurement errors. If we want to build a model to infer the relationship between a specific disease and the amount of alcohol consumption, while removing the effect of genes, we need to consider a high dimensional measurement error model. For examples of measurment error models and its application, see \citet{chen2005measurement}, \citet{fuller2009measurement}, \citet{schennach2016recent} and \citet{wang2023variable}. However, most studies on inference of measurement error models assume low dimensional settings with fixed covariate dimensions. As a result, there is still a paucity of methods for inferring measurement error models in high dimensional settings. \citet{li2021inference} proposed a corrected decorrelated score test for high dimensional measurement error models based on the linear model assumption. 

The validity of the statistical inference in high dimensional models always relies on explicit model assumptions. A lot of them are based on linear or generalized linear model assumptions. In this case, one stream of methods for statistical inference of high dimensional models are based on de-biasing, de-sparisifing or de-correlation of regularized M-estimators \citep{zhang2014confidence, van2014asymptotically, javanmard2014confidence}, and \citealt{ning2017general}). Another newly developed method is by partial penalizing on the parameter of interests \citep{shi2019linear}. However, in practice, the specific linear model assumption is always impractical and misspecified. The statistic based on model-misspecification has to tolerate the efficiency loss. Thus, it is of interest to develop an inference procedure that can accommodate potential misspecifications. 

The inference of high dimensional misspecified models without any measurement errors has been discussed in the literature. One commonly used approach is to construct the score function, that satisfies the double robust property. It means that this type of score function is locally insensitive to the estimation error, and can obtain a consistent estimator with asymptotical normality when either one of the main model or an ansillary model is correctly specified. \citet{chernozhukov2018double} introduced the general framework, and \citet{buhlmann2015high} demonstrated that the de-sparsified Lasso procedure proposed by \citet{van2014asymptotically} can be adapted to the inference of high dimensional misspecified models. However, these works relies on the sparsity conditions on both models, and need stronger conditions for the main models.   This work is closely related to the work under the approximate sparsity condition \citep{belloni2014inference}. This condition implies that a linear combination of a limited number of variables can approximate the unknown true mean function within a controllable error margin. In practice, the sparsity assumption is not checkable and often violated. 
In contrast, \citet{zhu2017breaking}, \citet{zhu2018linear} and \citet{zhu2018significance} proposed a valid inferential procedure by utilizing a Dantzing selector type estimator for dense coefficients.

Additionally, \citet{shah2023double} introduced a novel property called the "double-estimation-friendly" (DEF) property, which can be used to test for the significance of a variable $X$ for a response variable $Y$, controlling for other $p$ covariates $Z$. The DEF property implies that an inferential procedure is valid when either the model for $X$ regressed on $Z$, or the model for $Y$ regressed on $(X, Z^\top)^\top$, is correctly specified. \citet{shah2023double} demonstrated that their testing procedure can satisfy the DEF property, even for misspecified linear models and misspecified generalized linear models.

This article aims to develop a new inference procedure for testing the individual coefficient of a covariate that has measurement error and cannot be directly observed. In the high dimensional setting, in addition to the sparse setting, we explore a sparsity-adaptive procedure that allows for inference by relaxing the sparsity assumption. Our main contribution can be summarized into two aspects.

First, we propose a new inference procedure for one component coefficient that satisfies the DEF property, which is valid when either $X$ regressed on $\bZ$ is a sparse linear model or $Y$ regressed on $(X, \bZ^\top)^\top$ is a sparse linear model, allowing one of the two models to be misspecified. Specifically, we construct a double robust score function that permits the linear assumption of one of the models to be violated. We provide a theoretical guarantee for the asymptotic normality of the test statistic under mild conditions in both low-dimensional and high-dimensional settings. This method differs from the corrected score method \citep{ning2017general, bradic2020fixed}, as it also accounts for the measurement error in $X$. 

Second, we introduce the sparsity-adaptive procedure in which the sparsity assumption of one of the models can be violated, and we propose a methodology for inference under this type of violation in Section \ref{sec5.2.43/30}. Unlike the usual requirement for obtaining a consistent estimator, we only require that the estimator's gradient is close to zero. Under some mild conditions, we show that the asymptotic normality of the test statistic is still valid. 

The rest of this paper is organized as follows. Section \ref{sec5.23/30} presents the motivation and the formula of the double robust moment conditions. The correpsonding inference procedures for this type of misspecification are given in Section \ref{sec5.2.13/30} for low dimensional models and Section \ref{sec5.2.23/30} for high dimensional models. Section \ref{sec5.2.43/30} investigated the inference procedure by relaxing the sparsity condition in the high dimensional setting. We demonstrate the performance of our approach through simulation studies in Section \ref{sec5.33/30} and real data analysis in clinical trials in Section \ref{sec5.43/30}. A brief discussion is given in Section \ref{discussion}. Technical lemmas and proofs are presented in the Appendix.


Some notations are introduced before we move on. For a subgaussion random variable $S$, we define the subgaussian norm of $S$, $||S||_{\psi_2}$, as the smallest 
$K$ such that $E[exp(\lambda S)]\leq exp(\dfrac{\lambda^2K^2}{2})$ for all $\lambda$ holds. The sub-exponential norm of a random variable $X$ is denoted as $\|X\|_{\psi_{1}}=\sup _{q \geqslant 1} q^{-1}\left\{\mbE\left(|X|^{q}\right)\right\}^{1 / q}$
Note that $\|X\|_{\psi_{1}}<C_{1}$ for some constant $C_1$, if $X$ is sub-exponential. The sub-Gaussian
norm of $X$ is denoted as $\|X\|_{\psi_{2}}=\sup _{q \geqslant 1} q^{-1 / 2}\left\{\mbE\left(|X|^{q}\right)\right\}^{1 / q}$. Note that $\|X\|_{\psi_{2}}<C_{2}$ for
some constant $C_2$, if $X$ is sub-Gaussian.

\section{Test Statistics Construction}\label{sec5.23/30}
In this section, we introduce the model setting that incorporte the structure of measurement error and model-misspecification. To begin with, we introduce the notations: $[\bX,\bZ]\in \mathbb{R}^{n\times (p+1)}$ with $\bX=[X_1,\cdots,X_n]^\top\in\mathbb{R}^{n}$ and $\bZ=[\bZ_1^\top,\cdots,\bZ_n^\top]^\top\in\mathbb{R}^{n\times p}$ being the design matrix with sample size $n$. Each $\bZ_i=[Z_{i,1},\cdots,Z_{i,p}]^\top$ is a $p$-dimensional vector. $\bY=(Y_1,\cdots,Y_n)^\top\in\mathbb{R}^n$ is the response. Because of the measurement error, we cannot observe $\bX$. Instead, we observe $\bW$:
\begin{eqnarray}
\bW=\bX+\bU.
\end{eqnarray}
Here $\bU=[U_1,\cdots,U_n]^\top$ is the measurement error independent of $\bX,\bZ$ and $\bepsilon$. $\mbE(U_i)=0$ and $Var(U_i)=\sigma^2_U$ is assumed to be known.

Then, we consider the following general measurement error model setting:
\begin{eqnarray}\label{5.120220501}
	\begin{aligned}
Y \text{- Model:}\ \  Y_i &= X_i\beta+ g\left(\bZ_i\right) + \varepsilon_i,\\
X\text{- Model:}\ \ X_i & = f(\bZ_i) + \eta_i,
\   \end{aligned}
\end{eqnarray}
where  $\varepsilon_i \in\mathbb{R}$ is the error term satisfying $\mbE(\varepsilon_i)=0$, $Var(\varepsilon_i)=\sigma^2_\varepsilon$ and is assumed to be independent of $\bX$ and $\bZ$. Similarly, $\bfeta$ is random error independent of  $\bU$ and $\bepsilon$. We also assume that $\mbE(\eta_i|\bZ_i)=0$ and $Var(\eta_i)=\sigma_\eta^2$.  We are interested in testing 
\begin{equation}\label{null}
H_0: \beta=\beta^* \quad\text{v.s.}\quad H_1: \beta\neq\beta^*.
\end{equation}
Here $\beta^*$ is a specific value. 

In the previous literature, such as the debiased lasso \citep{javanmard2014confidence, zhang2014confidence, van2014asymptotically}, they required both $g(\bZ_i) = \bZ_i^{\top}\bgamma_0 $ and $f(\bZ_i) = \bZ_i^{\top}\btheta_0$ are true models for some sparse vectors $\bgamma_0$ and $\btheta_0$ given by
\begin{eqnarray*}
     \begin{aligned}
        &\bgamma_0 = \operatorname{argmin}\limits_{\bgamma} \mbE\left(Y_i - X_i\beta^* - \bZ_i^{\top}\bgamma\right)^2,\\
        &\btheta_0 = \operatorname{argmin}\limits_{\btheta} \mbE\left(X_i  - \bZ_i^{\top}\btheta\right)^2.
    \end{aligned}   
\end{eqnarray*}
In the work of \cite{zhu2018significance}, they allowed one of $\bgamma_0$ and $\btheta_0$ is a dense vector, but they still constrained both $Y$-model and $X$-model to the linear model framework. 

In our work, we relax the linear model assumption, and aim to construct test statistics that can be valid when only one of the $g(\bZ_i) = \bZ_i^{\top}\bgamma_0$ or $f(\bZ_i) = \bZ_i^{\top}\btheta_0$ is specified correctly. To be explicit, the inference procedure is valid when the true model is one of the following:
\begin{itemize}
\item[(a)]\label{amodel20220430}$Y_i=X_i\beta_0+\bZ_i^\top\bgamma_0+\varepsilon_i, W_i=X_i+U_i,  X_i=f(\bZ_i)+\eta_i.$ $f$ is an unknown  function. 
\item[(b)]\label{bmodel20220430}$Y_i=X_i\beta_0+g(\bZ_i)+\varepsilon_i,  W_i=X_i+U_i, X_i=\bZ_i^\top\btheta_0+\eta_i$. $g$ is an unknown function. 
\end{itemize}
The presence of two models is not only common in high-dimensional statistical inference but also has broad real-world applications. For instance, mediation analysis has been applied in various fields such as finance, psychology, communication research, genetics, and epidemiology. In mediation analysis, researchers investigate the mechanism through which an independent variable (predictor) affects a dependent variable (outcome) through a third variable, known as a mediator. This approach involves the introduction of both an outcome model and a mediation model. For further theory and applications, see \citet{mackinnon2012introduction}, \citet{baron1986moderator}, \citet{huang2015igwas}, \citet{huang2014joint} and \citet{cai2022new}.

Our method is motivated by the observation that when both (a) and (b) hold, the original parametric null hypothesis (\ref{null}) can be transformed into the following moment condition:
\begin{align}\label{5.520220501}
 \mbE\{\left(W_i-\bZ_i^{\top} \btheta_0\right)\left(V_i-\bZ_i^{\top} \bgamma_0\right)+\sigma_{U}^{2} \beta^{*}\}=0.
\end{align}
with $V_i :=Y_i-W_i\beta^*$ as a pseudo-response and $\bV=[V_1,\cdots,V_n]^\top$.  Moment condition ($\ref{5.520220501}$) are double-robust in the sense that if only one of the models from (a) and (b) holds, the condition is still satisfied.
Under the model (a) in which $\bZ_i^{\top}\btheta_0$ is not the true formula of $f\left(\bZ_i\right)$, the moment condition can be rewritten as
\begin{align}\label{dr: a}
    \mbE\{\left(U_i+\eta_i + f(\bZ_i) - \bZ_i^{\top}\btheta_0\right)\left(\varepsilon_i - U_i\beta^*\right) + \sigma_U^2\beta^*\} = 0.
\end{align}    
For the model (b), by similar process, we have
\begin{align}\label{dr: b}
  \mbE\{\left(U_i+\eta_i\right)\left(\varepsilon_i - U_i + g(\bZ_i) -\bZ_i^{\top}\bgamma_0\right) + \sigma_U^2\beta^*\} = 0.  
\end{align}
This double robust property implies that a valid inference procedure can be constructed based on this moment condition, requiring only one of the $X$-model or $Y$-model to satisfy a linear model assumption, while the other can be an arbitrary regression model. 

\begin{remark}
This moment condition is closely related to the corrected score method 
\citep{nakamura1990corrected} and the score decorrelation method \citep{ning2017general}. In the case where $\varepsilon_i$ is distributed as a Gaussian distribution and the model in equation (\ref{5.120220501}) is corrected, the corrected score equations for $\bgamma_0$ can be defined as:
\begin{align}
E[S_{\beta}(\beta^*, \bgamma_0)]&=E[W_i(V_i-\bZ_i^\top\bgamma_0)+\sigma_{U}^{2} \beta^{*}] = 0,
\\
E[S_{\bgamma}(\beta^*, \bgamma_0)]&=E[\bZ_i(V_i-\bZ_i^\top\bgamma_0)] = 0
\end{align}
up to a constant difference. The moment condition considers a corrected decorrelated score function  as
\begin{eqnarray}
S(\beta^*, \bgamma_0)=S_{\beta}(\beta^*, \bgamma_0)-\btheta_0^\top S_{\bgamma}(\beta^*, \bgamma_0).
\end{eqnarray}
Under $H_0$, we have $\mbE(S(\beta^*, \bgamma_0))=0$. This construction also satisfies $\mbE\left(S\left(\beta^*,\bgamma_0\right)S^{\top}_{\bgamma}\left(\beta^*,\bgamma_0\right)\right) = \bm{0}$, i.e., 
$$
\mbE\left(\bZ_i\left(W_i - \bZ_i^{\top}\btheta_0\right)\left(V_i - \bZ_i^{\top}\bgamma_0\right)^2\right)  = \bm{0}. 
$$
This motivates us to consider the moment condition (\ref{5.520220501}).
\end{remark}
\subsection{Low-Dimensional Setting}\label{sec5.2.13/30}

Based on (\ref{5.520220501}), we will construct an empirical test statistic based on the sample. Firstly, we start from the low-dimensional setting when  $p$ is fixed. To benefit the analysis, we introduce the following notations. Denote $\bSigma=E[(X_1, \bZ_1^T)^T(X_1, \bZ_1^T)], \bSigma_{X,X}=\mbE(X_1^2),$ $\bSigma_{Z,Z}=E[Z_1Z_1^T], \bSigma_{Z, W}=E[Z_1 W_1]$. $\bSigma_{Z, W}=\bSigma_{Z, X}$. Let $\lambda_{\max}(\bSigma)$ and $\lambda_{\min}(\bSigma)$ be the maximal and minimal eigenvalues of $\bSigma$.

Motivated by the moment condition (\ref{5.520220501}), consider the following test statistic

\begin{equation}\label{score}
T=\frac{1}{\sqrt n}\sum_{i=1}^n\{(W_i-\bZ_i^{\top}\hat{\btheta})(Y_i-W_i\beta^*-\bZ_i^{\top}\hat{\bgamma})+\sigma^2_U\beta^*\}.
\end{equation}\label{equa2_8}
Here $\hat{\btheta}$ and $\hat{\bgamma}$ are corresponding least squares estimators. To be specific,
\begin{eqnarray*}
&\hat{\btheta}=[\frac{1}{n}\sum_{i=1}^n \bZ_i\bZ_i^{\top}]^{-1}\frac{1}{n}\sum_{i=1}^n \bZ_iW_i=:\hat\bSigma_{Z,Z}^{-1}\hat\bSigma_{Z, W},\\
&\hat{\bgamma}=[\frac{1}{n}\sum_{i=1}^n \bZ_i\bZ_i^{\top}]^{-1}\frac{1}{n}\sum_{i=1}^n \bZ_i(Y_i-W_i\beta^*)=:\hat\bSigma_{Z,Z}^{-1}\hat\bSigma_{Z,Y-W\beta^*}.
\end{eqnarray*}

\begin{remark}
   It is noteworthy that $\bZ_i^{\top}\bgamma_0$ and $\bZ_i^{\top}\btheta_0$ should be considered as the best linear approximations of $g(\bZ_i)$ and $f(\bZ_i)$, respectively. Even if one of them is not the true function in our setting, obtaining the corresponding consistent estimators is still crucial to ensure the proposed procedure works. In fact, when one of the models in (\ref{5.120220501}) is misspecified lies in the fact that no matter whether $Y$-model or $X$-model holds or not, we have $\hat{\btheta}-\btheta_0=O_p(n^{-1/2})$ and $\hat{\bgamma}-\bgamma_0=O_p(n^{-1/2})$ under mild conditions which will be specified later. Specifically, suppose model (a) holds, $\btheta_0 = \bSigma_{\bZ,\bZ}^{-1}\mbE\left(\bZ_1W_1\right) = \bSigma_{\bZ,\bZ}^{-1}\mbE\left(\bZ_1f\left(\bZ_1\right)\right)$, it is straightforward to verify that $\hat{\btheta}$ is a root-n consistent estimator under assumption (A1) and (A2). Consequently, we can show $T = \frac{1}{\sqrt n}\sum_{i=1}^n\{(\varepsilon_i-U_i\beta^*)(X_i-\bZ_i^{\top}\btheta^*+U_i)+\sigma^2_U\beta^*\}+o_p(1)$.
\end{remark}

To consistently estimate asymptotic variance of $T$, we further define
\begin{equation}\label{score2}
\hat\sigma^2=\frac{1}{n}\sum_{i=1}^n\{(W_i-\bZ_i^{\top}\hat{\btheta})(Y_i-W_i\beta^*-\bZ_i^{\top}\hat{\bgamma})+\sigma^2_U\beta^*\}^2.
\end{equation}
where $\sigma^2 = \mbE\left(\left(W_1 - \bZ_1^{\top}\btheta_0\right)\left(Y_1 - W_i\beta^* - \bZ_1^{\top}\bgamma_0\right) + \sigma_U^2\beta^*\right)^2$. we 
Finally  define
\begin{equation}\label{5.1120220501}
T_{DF}=\frac{T}{\hat\sigma}=\frac{1}{\sqrt n\hat\sigma}\sum_{i=1}^n\{(W_i-\bZ_i^{\top}\hat{\btheta})(Y_i-W_i\beta^*-\bZ_i^{\top}\hat{\bgamma})+\sigma^2_U\beta^*\}.
\end{equation}

 The construction of $T_{DF}$ is similar to the Rao score test and decorrelated score test \citep{ning2017general}, except the term $\sigma_U^2\beta^*$. This term quantify the influence of measurement error $U$, and will disappear without the measurement error setting. 

To establish the asymptotic propoerty of $T_{DF}$, we introduce the following conditions
\begin{itemize}
    \item[(A1)] \label{subgau} $||\varepsilon_i||_{\psi_{2}}, ||U_i||_{\psi_{2}}, ||\eta_i||_{\psi_{2}}\left\|Z_{i,j}\right\|_{\psi_{2}}, \left\|X_{i}\right\|_{\psi_{2}},\|\bZ_i^T\btheta_0\|_{\psi_{2}}$ and $\|\bZ_i^T\bgamma_0\|_{\psi_{2}}$ are uniformly bounded by $K_1$ for $i=1,\cdots, n$ and $j=1,\cdots, p$. 
    \item[(A2)] $f$ is bounded on the support of $Z_1$ when (a) holds and $g$ is bounded on the support of $Z_1$ when (b) holds.
\end{itemize}

In Condition (A1), we assume the relevant random variables are all sub-Gaussian distributed with uniformly upper bounded Orlicz norm. This is commonly used in the literature of asymptotic analysis \citep{ning2017general,li2021inference}. It can further to ensure that the terms all have upper bounded fourth-order moment, which is important to check the Lyapunov's condition. Condition (A2) is to avoid extreme values of regression functions.

We have the following theoretical results:
\begin{theorem}\label{thm5.120220502}
Suppose (A1)-(A2) hold, under the null hypothesis, $T_{DF}\stackrel{d}{\longrightarrow} N(0,1)$ if  model (a) or (b) holds.
\end{theorem}

Theorem \ref{thm5.120220502} establishes the asymptotic normality of the newly proposed statistic. The asympototic normality holds when one of the models in (\ref{5.120220501}) is misspecified. Therefore we show that the statistic can have the DEF property. Different from existing studies, our procedure only requires knowledge of the variance of $U$, which is commonly adopted in the literature.


Next consider the following local alternative that the true value of $\beta$ is $\beta_0=\beta^*+\dfrac{c}{\sqrt n}$ for some $c>0$.
We have the following results:
\begin{theorem}\label{thm5.220220503}
Under the local alternative $\beta_0=\beta^*+\dfrac{c}{\sqrt n}$,
\begin{itemize}
\item When model (a) holds, it follows that $T_{DF}\rightarrow N\left(c\sigma^{-1}_1\sigma^2_{X,Z},1\right) $ in distribution, where  $\sigma^2_{X,Z}:=\bSigma_{X,X}-\bSigma^\top_{X,Z}\bSigma^{-1}_{Z,Z}\bSigma_{X,Z},  \sigma^2_1=E[(\varepsilon_1-U_1\beta^*)(f(Z_1)-Z^\top_1\btheta_0+\eta_1+U_1)+\sigma^2_U\beta^*]^2$.
\item When model (b) holds, it follows that $T_{DF}\rightarrow N\left(c\sigma^{-1}_2\sigma^2_{X,Z},1\right)$ in distribution, where $\sigma^2_2=E[\eta_1(\varepsilon_1-U_1\beta^*+g(Z_1)-Z^\top_1\bgamma_0)+\sigma^2_U\beta^*]^2$.
\end{itemize}
\end{theorem}


The above results imply that our proposed test statistic $T_{DF}$ can detect a local alternative with a rate of $n^{-1/2}$. It is worth noting when  $f(\bZ_i) = \bZ_i^{\top}\bgamma_0$ and $g(\bZ_i) = \bZ_i^{\top}\btheta_0$, we have $$\sigma^2_1=\sigma^2_2=E[\eta_1(\varepsilon_1-U_1\beta^*)+\sigma^2_U\beta^*]^2.$$

\subsection{High-Dimensional Setting}\label{sec5.2.23/30}

Next, let's turn to the high-dimensional measurement error model, where $p$ is allowed to be larger than $n$ and grow with $n$. 

 We still consider test statistic (\ref{5.1120220501}). However,  the least squared estimator is not valid when $p>n$. Instead, we can make the statistic adaptive for high-dimensional settings by switching $\hat{\bgamma}$ and $\hat{\btheta}$ in (\ref{5.1120220501}) to their corresponding penalized estimators. Namely,  we consider the penalized linear regressions:
\begin{eqnarray}
 \hat{\bgamma}&:=&\underset{ \bgamma \in \mathbb{R}^{p}}{\arg \min }\left\{\dfrac{1}{2n}\sum_{i=1}^n(Y_i-W_i\beta^* -\bZ_i^\top\bgamma)^2+p_{\lambda_{\bgamma}} (\bgamma)\right\},\label{5.1220220501}
\\
\hat{\btheta}&:=&\underset{\beta \in \mathbb{R}^{p}}{\arg \min }\left\{\dfrac{1}{2n}\sum_{i=1}^n(W_i-\bZ_i^\top \btheta)^2 +p_{\lambda_{\btheta}}(\btheta)\right\}.\label{5.1320220501}
\end{eqnarray}
The forms of $T$ and $T_{DF}$ are unchanged.  Let $p_\lambda(\cdot)$ be one of $p_{\lambda_{\btheta}}(\cdot)$ and $p_{\lambda_{\bgamma}}(\cdot)$. $p_\lambda$ should satisfy the following conditions:
\begin{enumerate}[label=(C\arabic*),ref=(C\arabic*),series=myexample2]
\item\label{(C1)} $p_\lambda(0) = 0$ and $p_\lambda$ is symmetric around zero.
\item For $t > 0$, the function $p_\lambda(t)/t$ is nondecreasing.
\item The function $p_\lambda(t)$ is differentiable for all $t\neq 0$ and subdifferentiable at $t = 0$, with $\lim_{t\rightarrow 0^+}=p_{\lambda}^{\prime}(t)=\lambda L$ for some $L>0$.
\item  There exists $\mu > 0$ such that $p_{\lambda,\mu}(t):=p_\lambda(t)+\dfrac{\mu}{2} t^{2}$ is convex.
\end{enumerate}

The above conditions are the same as the Assumption 1 in \citet{loh2015regularized} and are essential for showing the consistency of $\hat{\bgamma}$ and $\hat{\btheta}$ to $\bgamma_0$ and $\btheta_0$.
The Lasso, the SCAD and the MCP penalty all satisfy these conditions. $\lambda_{\bgamma}(\cdot)$ and $\lambda_{\btheta}(\cdot)$ are the corresponding tuning parameters for (\ref{5.1220220501}) and (\ref{5.1320220501}).   In Lemma 2,  we will check the conditions in \citet{loh2015regularized}, and show the consistency rate of the estimator under the proposed objective function.

In order for a similar inference procedure to be valid, a sparsity assumption is required on both $\bgamma_0$ and $\btheta_0$. The sparsity levels for $\btheta_0$ and $\bgamma_0$ need to satisfy the following condition.
\begin{enumerate}[resume*=myexample2]
\item\label{c.43/30} $s_{\bgamma}:=\|\bgamma_0\|_0=o(\sqrt{n}/\log p),s_{\btheta}:=\|\btheta_0\|_0=o(\sqrt{n}/\log p)$.
\end{enumerate}

This sparsity condition is commonly used in the high-dimensional estimation and inference problem, and is crucial to ensure the consistency of $\hat{\bgamma}$ and $\hat{\btheta}$ \citep{javanmard2014confidence, ning2017general, li2021inference}. In practice, the sparsity assumption is hard to check and might be violated. We will propose a methodology that can work when sparsity assumption of one of $\bgamma_0$ and $\btheta_0$ is violated in Section \ref{sec5.2.43/30}.

To show asymptotic normality  of $T_{DF}$, the following conditions are  also required.
\begin{enumerate}[resume*=myexample2]
\item\label{consigma} There exists some constant $\underline{C}>0$ such that $\lambda_{\min}(\bSigma)\geq\underline{C}$ uniformly.
\item \label{Cond6}When model (a) holds, let $\sigma^2_1=E[(\varepsilon_1-U_1\beta^*)(f(Z_1)-Z^\top_1\btheta_0+\eta_1+U_1)+\sigma^2_U\beta^*]^2$. When model (b) holds, let $\sigma^2_2=E[\eta_1(\varepsilon_1-U_1\beta^*+g(Z_1)-Z^\top_1\bgamma_0)+\sigma^2_U\beta^*]^2.$ We assume $\sigma_1, \sigma_2$ are uniformly bounded below by  $\sigma_0>0$. 
\end{enumerate}

Conditions (C6) and (C7) are used to ensure the variance $\sigma_{X,\bZ}^2$, $\sigma_1^2$, and $\sigma_2^2$ are bounded from below,  which guarantees that the test statistic is valid and possesses non-zero efficiency.

\begin{theorem}\label{thm5.420220502}
Assume conditions (A1) - (A2), (C1) - (C7) hold. If $\lambda_{\bgamma} = D_1\log p/n$ and $\lambda_{\btheta} = D_2\log p/n$, for $D_1$,$D_2$ sufficiently large, then it follows that under null hypothesis, $T_{DF}\rightarrow N(0,1)$ in distribution if  model (a) or (b) holds.
\end{theorem}
Theorem~\ref{thm5.420220502} establishs the asymptotic normality of the proposed method under the high-dimensional setting with suitable choice of the tuning parameters. It can also have local power under local alternatives. We have the following results:
\begin{theorem}\label{thm5.4420220502}
Under the same conditons with Theorem \ref{thm5.420220502},  we have that under the local alternative $\beta_0=\beta^*+\dfrac{c}{\sqrt n}$,
\begin{itemize}
\item When model (a) holds, it follows that
$T_{DF}\rightarrow N\left(c\sigma^{-1}_1\sigma^2_{X,Z},1\right).$ in distribution, where  $\sigma^2_{X,Z}:=\bSigma_{X,X}-\bSigma^\top_{X,Z}\bSigma^{-1}_{Z,Z}\bSigma_{X,Z},  \sigma^2_1=E[(\varepsilon_1-U_1\beta^*)(f(Z_1)-Z^\top_1\btheta_0+\eta_1+U_1)+\sigma^2_U\beta^*]^2$.
\item When model (b) holds, it follows that
$T_{DF}\rightarrow N\left(c\sigma^{-1}_2\sigma^2_{X,Z},1\right)$ in distribution, where $\sigma^2_2=E[\eta_1(\varepsilon_1-U_1\beta^*+g(Z_1)-Z^\top_1\bgamma_0)+\sigma^2_U\beta^*]^2$.
\end{itemize}
\end{theorem}
Similar to the low-dimensional setting, when the models (a) and (b) both hold, $\sigma_1=\sigma_2$.

Theorem \ref{thm5.420220502} together with Theorem \ref{thm5.4420220502} show that $T_{DF}$ is asymptotic normal under both the null and the local alternative hypothesis. It has $DEF$ property that it is robust to misspecification of $Y$-model and $X$-model.

By Theorem \ref{thm5.420220502} and Theorem~\ref{thm5.4420220502}, we could invert test results to construct confidence regions. We first compute 
$T_{DF}(\beta_t)$ to be the test statistic under 
\begin{equation}
H_0: \beta_0=\beta_t \quad\text{v.s.}\quad H_1: \beta_0\neq\beta_t
\end{equation}
The $(1-\alpha)$ confidence region 
$C_\alpha$ is given by 
$$C_{\alpha}:=\left\{\beta_t \in \mathbb{R}:\left|T_{DF}(\beta_t)\right| \leq   \Phi^{-1}(1-\alpha/2)\right\},$$ where $\Phi(\cdot)$ is the cumulative distribution function of the of a standard normal distribution.  Corollary \ref{coro5.120220502} shows the  asymptotic
validity of the confidence region. During our simulation and real data analysis, we used grid search to determine the region of $C_{\alpha}$.
\begin{corollary}\label{coro5.120220502}
Under the same conditions with Theorem \ref{thm5.420220502},
the confidence set $C_{\alpha}$ satisfies that

$$P\left(\beta_0 \in C_{\alpha}\right)=P\left(\left|T_{DF}(\beta_0)\right| \leq \Phi^{-1}(1-\alpha/2)\right) \rightarrow 1-\alpha.$$

\end{corollary}

\section{Sparsity-Adaptive Procedure}\label{sec5.2.43/30}
 
In high dimensional settings, the statistic proposed above is valid when Condition (C5) holds, which requires that the coefficients $\bgamma_0$ and $\btheta_0$ are both sparse enough. However, this assumption might not hold in practice. For instance, biologists have found examples from genome-wide gene expression profiling that all genes are believed to affect a common disease marker, thereby violating the sparsity assumption. Thus, we need to consider a methodology that can work when the sparsity assumption is violated, specifically when either $\bgamma_0$ or $\btheta_0$ is dense. \citet{zhu2017breaking}, \citet{zhu2018linear}, and \citet{zhu2018significance} have proposed methods for statistical inference in non-sparse high dimensional linear models. \citet{bradic2020fixed} have generalized these methods to linear mixed models. Instead of using the penalty-based estimator directly for inference, they modified the Dantzig selector estimator proposed by \citet{candes2007dantzig} and made it adaptive for inference in high dimensional models. We leverage their estimation procedure and generalize it to models with measurement error.

In this section, we consider the scenario where the sparsity assumption of either $\bgamma_0$ or $\btheta_0$ is violated, but we still assume that the linear model is valid. Specifically, we want the inference procedure to be valid when the true model is one of the following:
\begin{itemize}
\item[(c)]$\bY=\bX\beta_0+\bZ\bgamma_0+\bepsilon, \bW=\bX+\bU,  \bX=\bZ\btheta_0+\beta$,  $\|\btheta_0\|_0=o(\sqrt{\dfrac{n}{\log p\log n}})$,$\|\btheta_0\|_2\leq \kappa$ and $\|\bgamma_0\|_2\leq \kappa$ for some constant $\kappa$.
\item[(d)]$\bY=\bX\beta_0+\bZ\bgamma_0+\bepsilon, \bW=\bX+\bU,  \bX=\bZ\btheta_0+\beta$,  $\|\bgamma_0\|_0=o(\sqrt{\dfrac{n}{\log p\log n}})$,$\|\bgamma_0\|_2\leq \kappa$ and $\|\btheta_0\|_2\leq \kappa$ for some constant $\kappa$.
\end{itemize}

We still consider using statistics (\ref{5.1120220501}), but the difference here is that the penalty-based estimator cannot work well due to the loss of the sparsity assumption. Instead, we require a sparsity-adaptive estimation procedure. Recall that $\bV$ is the pseudo-response, given by $\bV:=\bY-\bW \beta_0$. Following a similar strategy used by \citet{bradic2020fixed}, we define an estimator that is adaptive to the level of sparsity based on the solution path of the following linear program. Let 
\begin{eqnarray}
\widetilde{\bgamma} & := & \underset{\bgamma \in \mathbb{R}^{p}}{\arg \min }\|\bgamma\|_{1}\label{5.1520220502} \\
& \text { s.t. } & \left\|n^{-1} \mathbf{Z}^{\top}(\mathbf{V}-\mathbf{Z} {\bgamma})\right\|_{\infty} \leq \eta_{\bgamma} , \label{5.1820220430}\\
& & \|\mathbf{V}-\mathbf{Z} {\bgamma}\|_{\infty} \leq \mu_{\bgamma}, \label{5.1920220430} \\
& & n^{-1} \mathbf{V}^{\top}(\mathbf{V}-\mathbf{Z} {\bgamma}) \geq \rho_{\bgamma},\label{5.2020220430}
\end{eqnarray}
with $\eta_{\bgamma}\asymp(\log n) \sqrt{n^{-1}\log p},$ $\mu_{\bgamma}\asymp\sqrt{n}$ and $0<\rho_{\bgamma}<\sigma^2_\varepsilon+\sigma_U^2\beta^{*2}$ as tunning parameters.
Compared with the Dantzig selector, there are two additional constraints (\ref{5.1920220430}) and (\ref{5.2020220430}). The constraint (\ref{5.1820220430}) guarantees the gradient of the loss is close to $0$. When $\bgamma_0$ is sparse, the constraint (\ref{5.1820220430}) is sufficient for the consistency of $\hat{\bgamma}$. $\hat{\bgamma}$ will be consistent of $\bgamma_0$ and share properties with the Dantzig selector. When $\bgamma_0$ is not sparse, the estimator will be no longer consistent. Under this case, additional constraints (\ref{5.1920220430}) and (\ref{5.2020220430}) can introduce the stability and make the final estimator still valid for inference.

Estimator of $\btheta_0$ is constructed in a similar way. For $\eta_{\btheta}>0$, let
\begin{eqnarray}
\widetilde{\btheta} &:= & \underset{{\btheta} \in \mathbb{R}^{p}}{\arg \min }\|{\btheta}\|_{1} \label{5.2120220502}\\
& \text { s.t.}& \left\|n^{-1} \mathbf{Z}^{\top}(\mathbf{W}-\mathbf{Z} {\btheta})\right\|_{\infty} \leq \eta_{\btheta}, \\
&& \|\mathbf{W}-\mathbf{Z} {\btheta}\|_{\infty} \leq\mu_{\btheta}, \\
&& n^{-1} \mathbf{W}^{\top}(\mathbf{W}-\mathbf{Z} {\btheta}) \geq \rho_{\btheta},\label{5.2420220502}
\end{eqnarray}
with $\eta_{\btheta}\asymp(\log n) \sqrt{n^{-1}\log p},$ $\mu_{\btheta}\asymp\sqrt{n}$ and $0<\rho_{\btheta}<\sigma^2_\eta+\sigma_U^2\beta^{*2}$.

Both optimization problems of solving $\widetilde{\bgamma}$ and $\widetilde\btheta$ can be formulated in the form of linear programming problems and the optimal can be obtained by the simplex method efficiently. We refer to Section 2.3 in \citet{zhu2018significance} as a reference for the optimization. 

The forms of $T$ and $T_{DF}$ remain unchanged. The following regularity conditions are required before we move on the the theoretical properties.
\begin{enumerate}[resume*=myexample2]
\item\label{(C10)20220501} $||\varepsilon_i||_{\psi_{2}}, ||U_i||_{\psi_{2}}, ||\eta_i||_{\psi_{2}}\left\|Z_{i,j}\right\|_{\psi_{2}}, \left\|X_{i}\right\|_{\psi_{2}}$are uniformly bounded by $K_1$ for $i=1,\cdots, n$ and $j=1,\cdots, p$. 
\item\label{(C12)20220501}  There exists constants $\bar{C},\underline{C}>0$ such that $\underline{C}\leq\lambda_{\min}(\bSigma)\leq\lambda_{\max}(\bSigma)\leq\bar{C}$ uniformly.
\item \label{C1220220502}  For $s=s_{\bgamma}$ when model (c) holds or $s=s_{\btheta}$  when model (d) holds, \\
$P(\min \limits_{J_{0} \subseteq\{1, \ldots, p\}\atop \left|J_{0}\right| \leq s}\min \limits_{\delta \neq 0\atop \|{\delta}_{J_{0}^{c}}\|_{1} \leq\|{\delta}_{J_{0}}\|_{1}}\dfrac{\left\| \bZ {\delta}\right\|_{2}}{\sqrt{n}\left\|{\delta}_{J_{0}}\right\|_{2}} \geq \kappa)\rightarrow 1$. 
\item \label{C1320220502}There exists a constant $\eta>0$ such that $\hat\sigma\stackrel{p}{\longrightarrow} \eta$.
\end{enumerate}
\emph{Remark.} Condition (C10) is also named as restricted eigenvalue condition. When the covaraites are sub-Gaussian distributed, \citet{rudelson2012reconstruction} show that the restricted eigenvalue condition holds with probability tending to 1 when sample size $n$ goes to infinity.

\begin{theorem}\label{them5.520220501}
 Suppose that (C8) - (C11) hold. Then under the null hypothesis, it holds that $T_{DF}\stackrel{d}{\longrightarrow} N(0,1)$ if  model (c) or (d) holds.
\end{theorem}

Theorem \ref{them5.520220501} establishes the asymptotic normality of our test statistics under the null hypothesis. The statistics only require one of $\btheta_0$ and $\bgamma_0$ to be sparse. The theoretical power function of our statistics will be in the same form as in Theorem \ref{thm5.420220502} when both $\btheta_0$ and $\bgamma_0$ are sparse, but it is hard to derive when the parameters are dense. Therefore, we leave it as future work. However, numerical studies indicate that the statistics can have power under the alternative hypothesis, as illustrated in Example 4 in Section \ref{sec5.33/30}.

\section{Simulation Study}\label{sec5.33/30}
We conduct simulation studies under different settings to investigate the finite sample performance of our test. In all simulation studies, the regression errors $\varepsilon_i$ and $\eta_i$ are generated from the standard normal distribution $N(0,1)$, and the measurement error $U_i$ is generated from $N(0,\sigma^2_U)$. The covariate vector $\bZ_i=(Z_{i,1},Z_{i,2},\cdots,Z_{i,p})^\top$ is generated from a multivariate Gaussian distribution $N_p(0,\bm{\bSigma})$, where $\bm{\bSigma}$ is the autoregressive matrix with its entry $\bSigma_{j,k}=\rho^{|j-k|}$. We set the nominal level for rejection to be 0.05.

\subsection{Performance under Low Dimensional Setting}
\textbf{Example 1.} We first consider the low dimensional setting with $p=4$ and sample size $n=100$. We set $\rho=0.5$ and $\sigma_U=1$. Simulation 1 represents the case when the $Y$-model is misspecified, and Simulation 2 represents the case when the $X$-model is misspecified.

\textit{Simulation 1.} In Simulation 1, let $X$-model  be misspecified. We generate $X_i=Z_{i,1}^2+\eta_i$, $X_i=\text{sin}(Z_{i,1})+\text{sin}(Z_{i,2})+\eta_i$ and $X_i=Z_{i,1}Z_{i,2}+\eta_i$ in Simulation $1a$, $1b$ and $1c$ respectively. Let $W_i=X_i+U_i$ and    $Y_i=X_i\beta+Z_{i,1}+0.8Z_{i,2}+\varepsilon_i$. 

\textit{Simulation 2.} In Simulation 2, let $Y$-model  to be misspecified while $X$-model is correctly specified. We generate $X_i=Z_{i,1}+0.8Z_{i,2}+\eta_i$, and $W_i=X_i+U_i$. Here we consider three different types of regression functions of the response variable.  Let $Y_i$ to be generated from $Y_i=X_i\beta+Z_{i,1}^2+\varepsilon_i$, $Y_i=X_i\beta+\text{sin}(Z_{i,1})+\text{sin}(Z_{i,2})+\varepsilon_i$ and $Y_i=X_i\beta+Z_{i,1}Z_{i,2}+\varepsilon_i$ in Simulation $2a$, $2b$ and $2c$ respectively.

Consider the following null and alternative hypothesis:
\begin{equation}
H_{0}: \beta=0\text{\ \ versus\ \ } H_{1}:\beta \neq 0.
\end{equation}
Let the true value of $\beta$ vary from $-2$ to $2$. It corresponds to null hypothesis when $\beta= 0$ while it corresponds to alternative hypothesis when $\beta\neq0$.

The simulation results are presented in Figure \ref{sim12}.
As expected, the newly proposed statistics control the size well and is powerful under different kinds of model misspecification.

\begin{figure}[htbp]
\centering    
\subfigure 
{
	\begin{minipage}[b]{0.45\textwidth}
	\centering          
	\includegraphics[width=5.8cm]{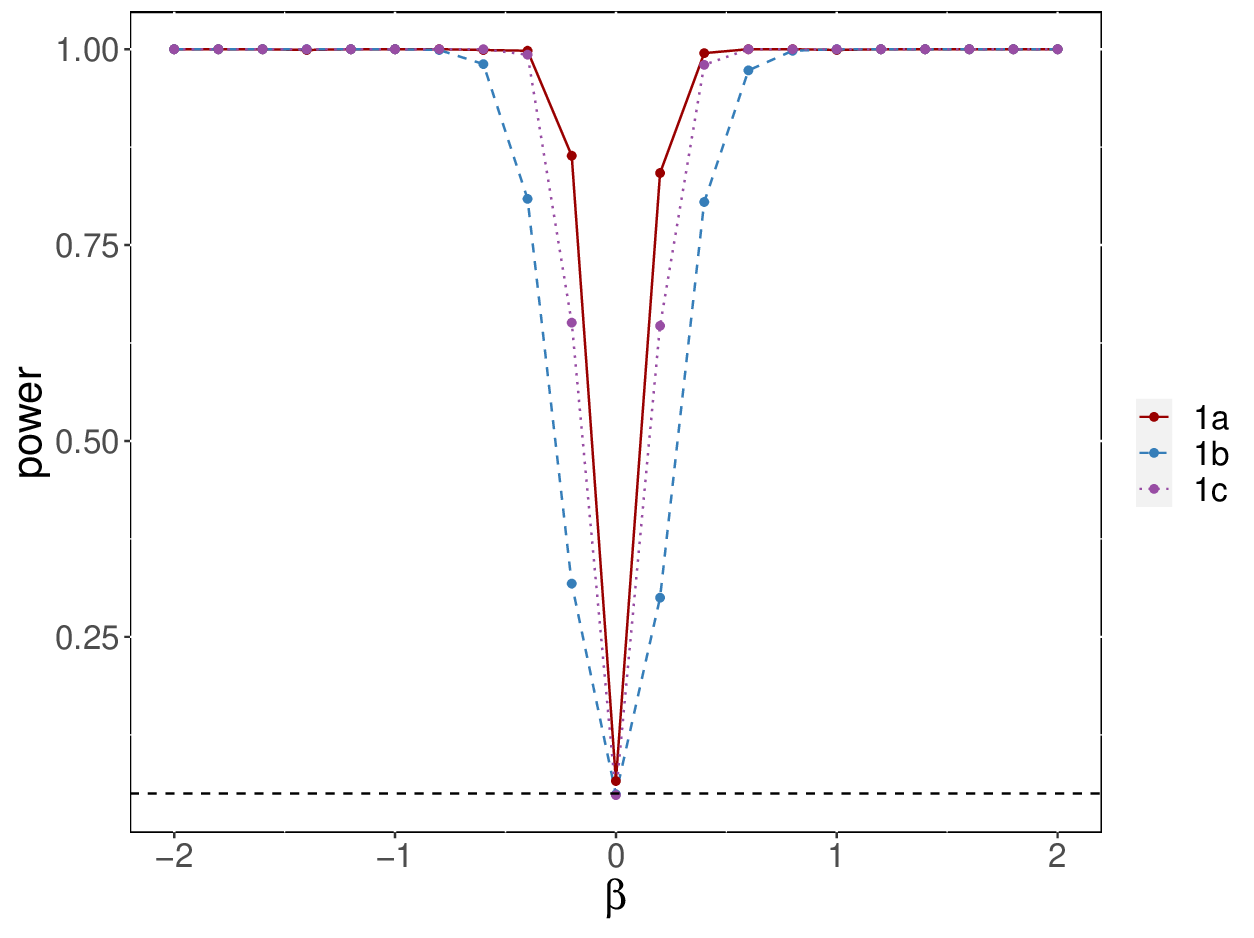}   
	\end{minipage}
}
\subfigure 
{
	\begin{minipage}[b]{0.45\textwidth}
	\centering      
	\includegraphics[width=6.3cm]{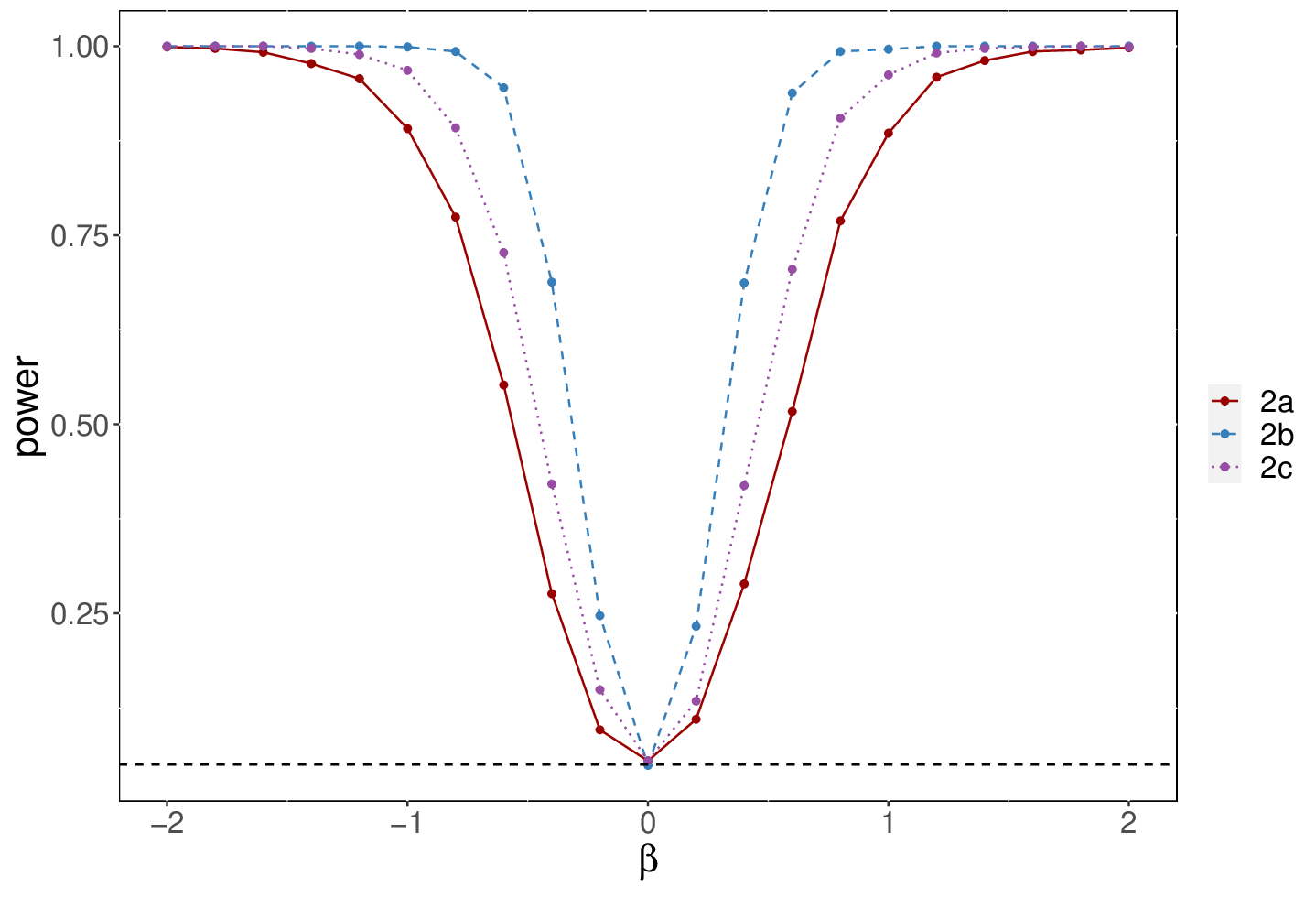} 
	\end{minipage}
}
\caption{Left panel and right panel are the empirical sizes and powers of $T_{DF}$ at level $\alpha = 0.05$ over 1,000 replications
in Simulation 1 and 2 respectively. The horizontal dotted line represents significance level 0.05}\label{sim12}
\end{figure}

\subsection{Performance under High Dimensional Setting}
We then consider high dimensional setting in the following examples. In the following examples, sample size $n$ is set to $200$. In Example 1 and Example 2, we show the performance of $T_{DF}$ under model misspecification with the parameters still being sparse. In these examples, we use the Lasso estimator. We use $T_{DFLA}$ to represent our statistics with the lasso estimator. In Example 3, the performance of $T_{DF}$ using the sparsity-adaptive estimator is shown, under models with dense parameters.  Use $T_{DFDE}$ to represent the statistics with the sparsity-adaptive estimator.
We compare the performance of the newly proposed test statistics with the corrected decorrelated score test (CDST) statistics proposed by \citet{li2021inference}. 

\textbf{Example 2.}
In this example, we consider the case when both $X$-model and $Y$-model are correctly specified.  We take $p=200$, $\rho=0.25$ and $\sigma_U=0.1$. We generate $X_i=1.2Z_{i,1} + 0.8Z_{i,2}+ \eta_i$, $Y_i= X_i\beta + Z_{i,3} + Z_{i,4} + \varepsilon_i$ and $W_i=X_i+U_i$. 
The null and alternative hypothesis of this example is:
\begin{equation}\label{sim1}
H_{0}: \beta=1\text{\ \ versus\ \ } H_{1}:\beta \neq 1.
\end{equation}
Simulation results are shown in the left penal in Figure \ref{fig:1}. CDST performs similar to our method under this setting. 

CDST cannot control the size when predictors are highly correlated with each other. We further adjust the previous $X$-model to $X_i=1.2Z_{i,1} + 0.8Z_{i,3}+ \eta_i$, and remain the  $Y$-model as the same. where $X_i$ and active predictors in $\bZ_i$ are correlated. Results are illustrated in the right penal in Figure \ref{fig:1}.
We can see from the plot that the type-I error of CDST is 0.13 under this setting, which indicates that CDST cannot control the size well. Further, CDST needs an estimation of the fourth-order moment of the measurement error so the computational burden is large and the condition on $U$ is more strict. The newly proposed method only requires and estimates the second-order moment of $U$, which is pretty common and moderate when analyzing the measurement error model.

\begin{figure}[htbp]
\centering    
\subfigure 
{
	\begin{minipage}[b]{0.47\textwidth}
	\centering          
	\includegraphics[width=5.8cm]{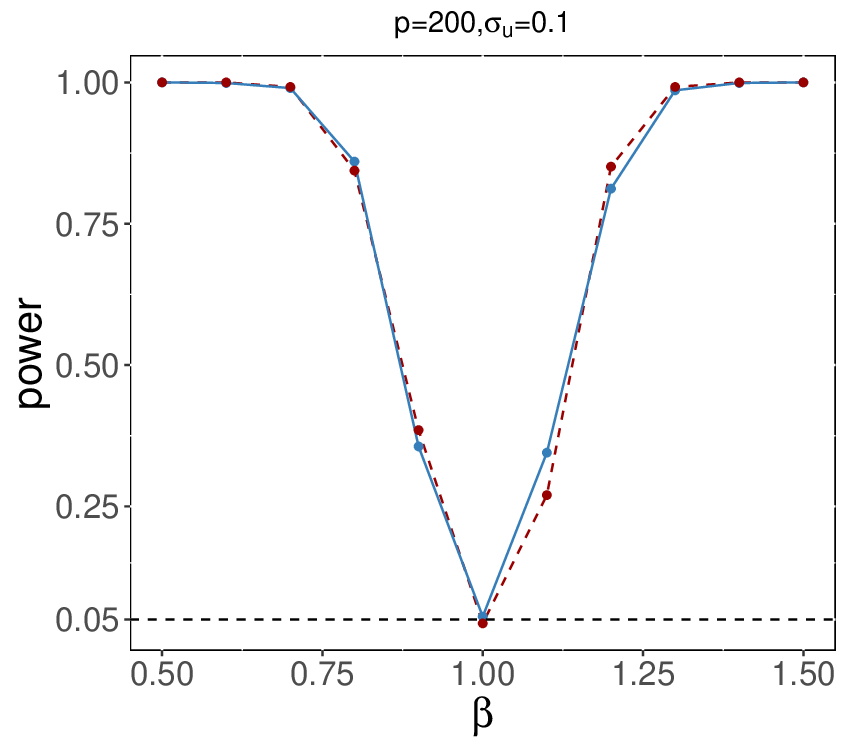}   
	\end{minipage}
}
\subfigure 
{
	\begin{minipage}[b]{0.47\textwidth}
	\centering      
	\includegraphics[width=6.9cm]{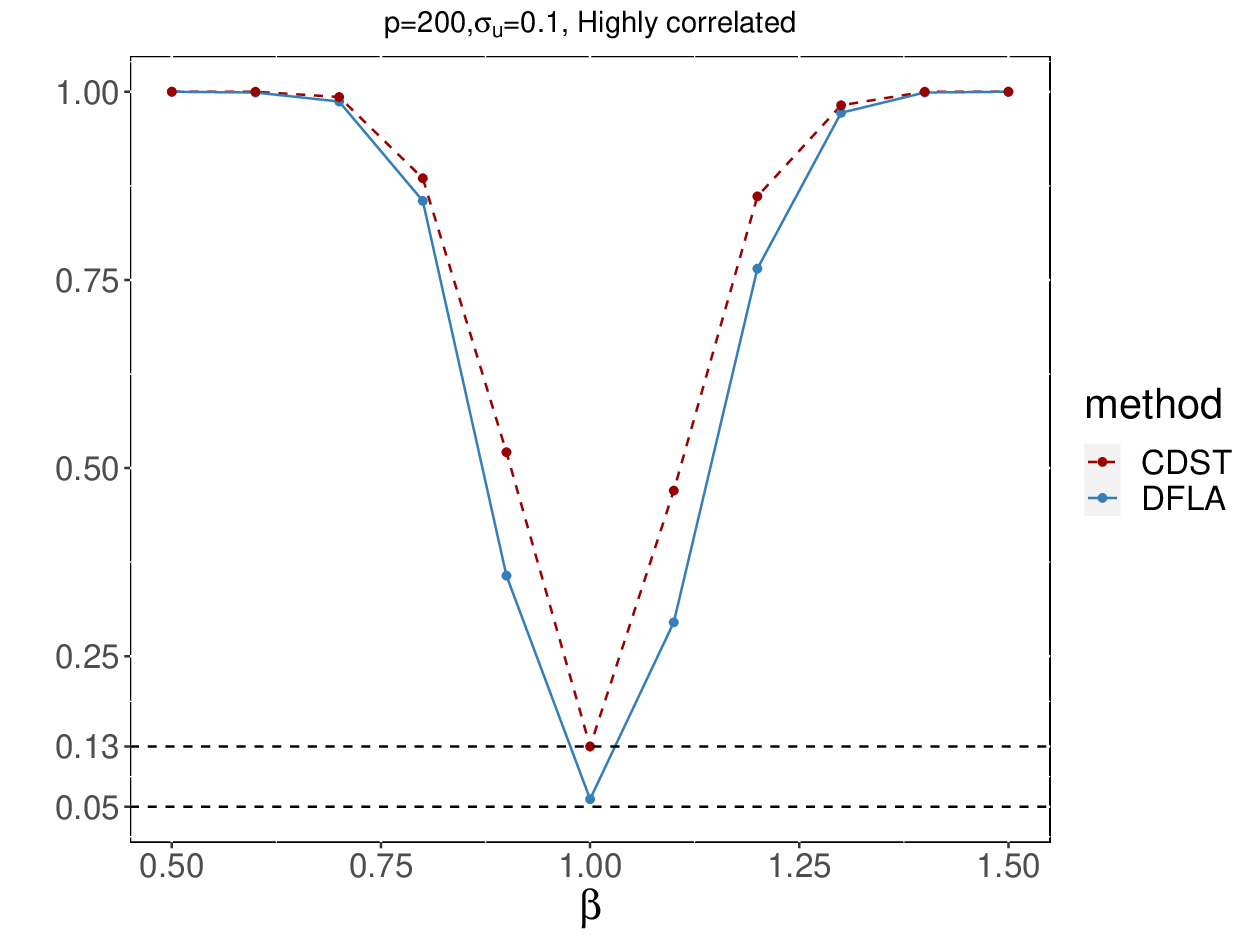} 
	\end{minipage}
}
\caption{Power comparison plot. Left: when $X$ and active $Z$ are not correlated. Right: when $X$ and active $Z$ are highly correlated}\label{fig:1} 

\end{figure}

\textbf{Example 3.}
We further investigate the performance of our test statistics when the $Y$-model or $X$-model is misspecified under a high dimensional setting. We consider three different models.

\textit{Model 1.}
In this model we let the $X$-model  be misspecified. Let $\widetilde{Z}_{i,j}=\dfrac{e^{Z_{i,j}}-1}{e^{Z_{i,j}}+1}$. We generate $X_i=2\widetilde{Z}_{i,1} + \widetilde{Z}_{i,2}+ \eta_i$, $Y_i= X_i\beta + Z_{i,3} + Z_{i,4} + f(X_i)\cdot\varepsilon_i$ and $W_i=X_i+U_i$. 
We let $p$ to be $200$ and $400$. $\rho=0.25$. Different values of $\sigma_U$ are also considered, where $\sigma_U= 0.1, 0.2,$ and $0.3$ respectively. With larger dimension or larger measurement error, inference becomes progressively more difficult. We consider three scenarios for $f(X)$ that (i): $f(X_i)=1$, (ii): $f(X_i)=X_i$ and (iii): $f(X_i)=1 + 0.5\sin(4\pi X_i)$, where error is heteroscedastic under the second and third scenarios. Simulation results are shown in Table \ref{model1sum1}.   The newly proposed statistics can control the size well under different settings while size of CDST deviates largely from 0.05 when models are misspecified.  There is little power loss when the dimension  and measurement error become larger.

\begin{table}[h!]
\caption{Empirical sizes and powers for Model 1}\label{model1sum1}
\centering
\begin{tabular}{ccccccccccc}
\toprule
\multirow{3}{*}{Model}           & \multirow{3}{*}{$p$}                   & \multirow{3}{*}{$\sigma_U$}& \multicolumn{7}{c}{$\beta$} \\ && & \multicolumn{2}{c}{$1$(size)}& $1.2$ & $1.4$ & $1.6$  & $0.8$ & $0.6$ & $0.4$ \\\cline{4-11} &&&$T_{DFLA}$&CDST&\multicolumn{6}{c}{------------------------$T_{DFLA}$-------------------------}\\ \hline
\multirow{6}{*}{$1(i)$} & \multirow{3}{*}{200} & 0.1& 0.057&0.440&0.793&1&1&0.859&0.999&1 \\ 
                   &                      & 0.2    &       0.063&0.473&0.75&0.999&1&0.855&0.999&1  \\  
                   &                      & 0.3    &       0.065&0.434&0.654&0.991&1&0.843&0.999&1  \\ \cline{2-11} 
                   & \multirow{3}{*}{400} & 0.1    &       0.044&0.432&0.795&1&1&0.843&1&1   \\  
                   &                      & 0.2    &     0.053&0.430&0.743&0.999&1&0.85&1&1  \\ 
                   &                      & 0.3    &      0.065&0.424&0.64&0.992&0.999&0.845&1&1   \\ \toprule

\multirow{6}{*}{$1(ii)$} & \multirow{3}{*}{200} & 0.1& 0.057&0.535&0.287&0.775&0.983&0.325&0.819&0.975 \\ 
                   &                      & 0.2    &       0.053&0.541&0.27&0.744&0.975&0.341&0.815&0.976   \\  
                   &                      & 0.3    &       0.053&0.503&0.239&0.699&0.959&0.35&0.813&0.977  \\ \cline{2-11} 
                   & \multirow{3}{*}{400} & 0.1    &       0.038&0.551&0.282&0.79&0.984&0.325&0.831&0.987   \\  
                   &                      & 0.2    &     0.035&0.502&0.254&0.761&0.979&0.344&0.831&0.988  \\ 
                   &                      & 0.3    &      0.041&0.549&0.205&0.699&0.954&0.351&0.834&0.989  \\ \toprule
\multirow{6}{*}{$1(iii)$} & \multirow{3}{*}{200} & 0.1& 0.06&0.405&0.762&0.997&1&0.821&0.999&1 \\ 
                   &                      & 0.2    &       0.059&0.445&0.721&0.995&1&0.814&1&1   \\  
                   &                      & 0.3    &       0.069&0.429&0.629&0.983&0.999&0.806&0.999&1  \\ \cline{2-11} 
                   & \multirow{3}{*}{400} & 0.1    &       0.044&0.420&0.766&0.999&1&0.816&1&1   \\  
                   &                      & 0.2    &     0.049&0.427&0.699&0.996&1&0.823&1&1 \\ 
                   &                      & 0.3    &      0.066&0.435&0.595&0.985&1&0.825&0.999&1   \\ \toprule

\end{tabular}
\end{table}

In the next part of simulation, We consider the case when the $X$-model is correctly specified to be sparse linear model while $Y$-model is misspecified.  We consider two models. 

\emph{Model 2}. In this model We set $(X_i,\bZ_i^\top)^\top\sim N(0,\bSigma),$ with $\bSigma_{i,j}=0.75^{|i-j|}$, which follows an AR(1) processes. This is the case when the sparsity level of $X$-model is 1. We consider two different models for $Y$:

\emph{Model 2a}. $Y_i = X_i \beta + \widetilde{Z}_i^\top\bgamma+ \varepsilon_i, \widetilde{Z}_{i,j}=\dfrac{ e^{Z_{i,j}}-1}{e^{Z_{i,j}}+1}.$ $\bgamma_1,\cdots,\bgamma_{10}$ are from Uniform(0,1). 

\emph{Model 2b}. $Y_i= X_i\beta + Z_{i,3}^2 + Z_{i,4}^2 + \varepsilon_i$.

Compared with model $2a$, model $2b$ is more challenging because the $Y$-model in model $2b$ is  highly nonlinear. Results are shown in Table \ref{model1sum} and  we can find that the newly proposed test statistics can still control the size well even under these more challenging scenarios. It is not surprising that power is not as large as before, especially when the measurement error becomes larger. Nevertheless, the power is still in a reasonable magnitude and goes to 1 when true value of $\beta$ deviates from the null setting.

\begin{table}[h!]
\caption{Empirical sizes and powers for Model 2}\label{model1sum}
\centering
\setlength{\tabcolsep}{0.8mm}
\begin{tabular}{ccccccccccccc}
\toprule
\multirow{3}{*}{Model}           & \multirow{3}{*}{$p$}                   & \multirow{3}{*}{$\sigma_U$}& \multicolumn{9}{c}{$\beta$} \\&& & \multicolumn{2}{c}{$1$(size)} & $1.2$ & $1.4$ & $1.6$ & $1.8$ & $0.8$ & $0.6$ & $0.4$ & $0.2$\\\cline{4-13}
&&&$T_{DFLA}$&CDST&\multicolumn{8}{c}{---------------------------$T_{DFLA}$----------------------------}\\\hline
\multirow{6}{*}{$2a$} & \multirow{3}{*}{200} & 0.1& 0.051&0.411&0.313&0.896&0.999&1&0.498&0.95&0.999&1 \\ 
                   &                      & 0.2    &       0.062&0.402&0.248&0.845&0.987&1&0.493&0.939&0.999&1   \\  
                   &                      & 0.3    &       0.075&0.380&0.17&0.697&0.957&0.999&0.486&0.924&0.999&1  \\ \cline{2-13} 
                   & \multirow{3}{*}{400} & 0.1    &      0.053&0.435&0.309&0.888&0.999&1&0.482&0.949&1&1   \\  
                   &                      & 0.2    &     0.063&0.450&0.24&0.809&0.995&1&0.491&0.943&1&1  \\ 
                   &                      & 0.3    &      0.08&0.410&0.155&0.659&0.949&0.996&0.508&0.924&0.999&1  \\ \toprule
\multirow{6}{*}{$2b$} & \multirow{3}{*}{200} & 0.1& 0.055&0.214&0.104&0.235&0.452&0.625&0.107&0.256&0.453&0.677 \\ 
                   &                      & 0.2    &       0.061&0.198&0.096&0.217&0.407&0.596&0.099&0.253&0.433&0.638   \\  
                   &                      & 0.3    &       0.063&0.194&0.095&0.198&0.355&0.538&0.098&0.231&0.407&0.603  \\ \cline{2-13} 
                   & \multirow{3}{*}{400} & 0.1    &      0.049&0.234&0.081&0.245&0.469&0.67&0.095&0.215&0.43&0.645  \\  
                   &                      & 0.2    &     0.045&0.227&0.07&0.221&0.436&0.639&0.1&0.214&0.419&0.62  \\ 
                   &                      & 0.3    &      0.043&0.250&0.063&0.187&0.389&0.578&0.097&0.21&0.388&0.577  \\ \toprule
\end{tabular}
\end{table}

\textbf{Example 4.}
In this example, we consider the case when the sparsity assumption in the $Y$-model is violated. $\bX$ and $\bZ$ are still generated as in Model 2 in Example 3. We generate a dense $Y$-model: $Y_i = X_i \beta + \bZ_i^\top\bgamma + \varepsilon_i,$ where $\bgamma_i=F_in^{-1/2}\mathbf{1}_p$. $F_i$ are generate from uniform distribution with support $(0, 1)$. True value of $\beta$ is set to be $1$. Results are shown in Table \ref{model2sum}. In the table, $T_{DFLA}$ is the statistics with the lasso estimator, and $T_{DFDE}$ is the statistics with the sparsity-adaptive estimator. From the table, we can see that $T_{DFDE}$ can control the size well and can have a high power under the alternative. Neither $T_{DFLA}$ nor CDST can control the size well because these two comparative methods depend on the assumptions that the main model is sparse and that we can correctly recognize the active variables. However, these assumptions are violated when the main model is dense.

\begin{table}[h!]
\caption{Empirical sizes and powers for Example 4}\label{model2sum}
\centering
\setlength{\tabcolsep}{1mm}{
\begin{tabular}{ccccccccccccc}
\toprule
 \multirow{3}{*}{$p$}                   & \multirow{3}{*}{$\sigma_U$}& \multicolumn{11}{c}{$\beta$} \\ &&  \multicolumn{3}{c}{1 (size)} & $1.2$ & $1.4$ & $1.6$& 1.8  & $0.8$ & $0.6$ & $0.4$ &0.2\\ \cline{3-13}
&&$T_{DFDE}$&$T_{DFLA}$&CDST&\multicolumn{8}{c}{---------------------------$T_{DFDE}$----------------------------}\\
\hline
  \multirow{3}{*}{200}                   & 0.1    &     0.055&0.105&0.299&0.160&0.348&0.525&0.774&0.093& 0.158&0.278&0.521\\ 
                                         & 0.2    &     0.057&0.184&0.292&0.124&0.349&0.501&0.742&0.085&0.129&0.305&0.562 \\  
                                         & 0.3    &     0.069&0.206&0.243&0.114&0.252&0.453&0.694&0.073&0.114&0.247&0.459\\ \cline{2-13} 
                   \multirow{3}{*}{400} & 0.1    &      0.055&0.133&0.299&0.159&0.368&0.532&0.782&0.099&0.159&0.298&0.523 \\  
                                         & 0.2    &     0.064&0.223&0.287&0.136&0.310&0.522&0.758&0.083&0.136&0.294&0.518  \\ 
                                         & 0.3    &     0.069&0.276&0.225&0.102&0.281&0.508&0.703&0.074&0.105&0.242&0.449\\ \toprule
\end{tabular}}
\end{table}

\section{Real Data Analysis}\label{sec5.43/30}
In this section, we illustrate the proposed methodology by applying it to a dataset in a clinical trial analyzed by \citet{chu2016feature} and \citet{li2021inference}. The data are collected from the Childhood Asthma Management Program (CAMP), which is a 4-year clinical trial aimed at finding the impact of daily asthma medications on lung development in growing children.

We follow a similar data preprocessing strategy as \citet{li2021inference}. The dataset consists of $n=273$ subjects, each of whom contributed up to 3 consecutive clinical visits in 6 months. The response variable $Y_i$ is the average asthma symptoms among the 3 visits. The FEV1/FVC ratio, which is the ratio of the forced expiratory volume in the first one second to the forced vital capacity of the lungs and believed to be related to asthma symptoms in biology, is the unobserved variable of interest denoted by $X_i$. However, FEV1/FVC is measured in a pulmonary function test but has potential measurement error, see \citet{belzer2019practical} for a reference. Therefore, we assume $U_{ij}$ to be the measurement error from each individual in each visit. We let $W_{ij}=X_{i}+U_{ij}$ be the observed value, and we calculate the average of three measurements to obtain $W_i=\frac{1}{3}\sum_{j=1}^3 W_{ij}$ and $U_i=\frac{1}{3}\sum_{j=1}^3 U_{ij}$. We estimate $\sigma_U^2$ to be one-sixth of the sample variance of $W_{i k}-W_{i j}$, which yields an estimate of 0.529.

The problem of interest is to understand how FEV1/FVC ratio, with potential measurement error, together with age, gender, and 676 SNPs would affect asthma symptoms. There are originally a total of eight hundred and seventy thousand SNPs, but we select 676 SNPs based on minor allele frequency (MAF). We assume that having two of the minor alleles has twice the effect on the phenotype as having one of the minor alleles, and zero means no effect. Therefore, we keep SNPs with higher MAF.

Since we have 679 covariates but only 273 observations, the data is high dimensional. We perform the proposed testing procedure with parameters estimated by the Lasso and sparsity-adaptive procedures. The results are shown in Table \ref{tab5.405310508}. $H_0$ in Table \ref{tab5.405310508} is the null hypothesis $H_0: \beta_0=0$. The confidence region is obtained by inverting the testing results using the method shown in Corollary~\ref{coro5.120220502}. All three methods reject $H_0$ and always give negative confidence regions, suggesting that the lower the FEV1/FVC ratio, the more severe the asthma in the patient.

\begin{table}[]
\centering
\caption{Results of Real Data}\label{tab5.405310508}
\begin{tabular}{cccc}
\toprule
    &Statistic under $H_0$ & $p$-value & Confidence region \\\hline
$T_{DFLA}$&   -4.15   &$1.63\times10^{-5}$ &$(-0.19,-0.35)$                 \\
$T_{DFDE}$&     -4.35    &   $6.81\times 10^{-6}$     &    $(-0.20,-0.40)$                  \\
$CDST$    &  5.15   &     $1.30\times10^{-7}$    &             $(-0.16, -0.73)$\\\toprule       
\end{tabular}
\end{table}

\section{Discussion}\label{discussion}
This work proposes a novel inference method of the single component of the unobserved covariate under the high dimensional model-misspecification scenario. The proposed method remains valid even when the sparsity condition is relaxed, and it does not require a consistent estimator. Both theoretical and numerical studies demonstrate that this approach effectively controls type I error and achieves non-trivial power in practice. The method relies on an additive model setting and requires the separability of $X\beta$ and $f(\bZ)$. It is of interest to consider a more general setting
$$
\mbE\left(Y|X,\bZ\right) = \mu\left(X\beta+f(\bZ)\right),
$$
where $\mu(\cdot)$ is some known nonlinear function. This framework contains models such as logistic regression, which is widely used in genetic and classification studies. Identifying a double robust moment condition when $W=X+U$ is observed instead of $X$ is beyond the scope of this paper. We will study it in the future.

\bibliographystyle{agsm} 
\bibliography{hdtest.bib}

\newpage
\section*{Appendix}\label{sec6tec}


\subsection*{Proof of Theorem \ref{thm5.120220502}}

Firstly, assume that Model (a) holds. In this situation, we have
\begin{eqnarray*}
&&Y_i-W_i\beta^*-\bZ_i^{\top}\hat{\bgamma}=\varepsilon_i-U_i\beta^*+\bZ_i^{\top}(\bgamma_0-\hat{\bgamma});\\
&&W_i-\bZ_i^{\top}\hat{\btheta}=X_i-\bZ_i^{\top}\btheta_0+U_i-\bZ_i^{\top}(\hat{\btheta}-\btheta_0).
\end{eqnarray*}

As a result, we get:
\begin{eqnarray*}
T&=&\frac{1}{\sqrt n}\sum_{i=1}^n\{(\varepsilon_i-U_i\beta^*)(X_i-\bZ_i^{\top}\btheta_0+U_i-\bZ_i^{\top}(\hat{\btheta}-\btheta_0))+\sigma^2_U\beta^*\}\\
&&+\frac{(\bgamma_0-\hat{\bgamma})^\top}{\sqrt n}\sum_{i=1}^n\bZ_i(W_i-\bZ_i^{\top}\hat{\btheta})\\
&=&\frac{1}{\sqrt n}\sum_{i=1}^n\{(\varepsilon_i-U_i\beta^*)(X_i-\bZ_i^{\top}\btheta_0+U_i)+\sigma^2_U\beta^*\}\\
&&-\frac{1}{\sqrt n}\sum_{i=1}^n(\varepsilon_i-U_i\beta^*)\bZ_i^{\top}\times(\hat{\btheta}-\btheta_0)\\
&=&\frac{1}{\sqrt n}\sum_{i=1}^\{[(\varepsilon_i-U_i\beta^*)(X_i-\bZ_i^{\top}\btheta_0+U_i)+\sigma^2_U\beta^*\}+o_p(1).
\end{eqnarray*}
The second equation is true due to the definition of the least squared estimator $\hat{\btheta}$. The last equation holds since $E[(\varepsilon-U\beta^*)Z]=0$ and the consistency of $\hat{\btheta}$ to $\btheta_0$. 

Note that 
\begin{eqnarray*}
&&E[(\varepsilon_i-U_i\beta^*)(X_i-\bZ_i^{\top}\btheta_0+U_i)+\sigma^2_u\beta^*]\\
&=&E[(\varepsilon_i-U_i\beta^*)(X_i-\bZ_i^{\top}\btheta_0)+\varepsilon_iU_i+(\sigma^2_U-U^2_i)\beta^*]=0.
\end{eqnarray*}
The last equation holds since $\mbE(\varepsilon|X,Z)=0$ and $E[U|X,Z,\varepsilon]=0$. 
Thus by the central limit theorem,
$$T\rightarrow N\Big(0, E[(\varepsilon-U\beta^*)(X-Z^\top\btheta_0+U)+\sigma^2_u\beta^*]^2\Big),$$
in distribution.

Further note that
\begin{eqnarray*}
\hat\sigma^2&=&\frac{1}{n}\sum_{i=1}^n\left\{[\varepsilon_i-U_i\beta^*+\bZ_i^{\top}(\bgamma_0-\hat{\bgamma})][X_i-\bZ_i^{\top}\btheta_0+U_i-\bZ_i^{\top}(\hat{\btheta}-\btheta_0)]
+\sigma^2_U\beta^*\right\}^2\\
&=&\frac{1}{n}\sum_{i=1}^n\{(\varepsilon_i-U_i\beta^*)(X_i-\bZ_i^{\top}\btheta_0+U_i)+\sigma^2_U\beta^*\}^2+o_p(1).
\end{eqnarray*}
Thus $\hat\sigma^2$ is a consistent estimator of corresponding asymptotic variance of $T$. We then can conclude that $T_{DF}$ converges to standard normal distribution.

\vspace{3mm}

Secondly, assume that Model (b) holds. Assume now under null hypothesis the true model is
\begin{equation}\label{mis}
Y_i=X_i\beta^*+f(Z_i)+\varepsilon_i.
\end{equation}
Here $f(Z_i)\neq Z_i^\top\bgamma$ for any $\bgamma$.

In this situation, we have
\begin{eqnarray*}
&&Y_i-W_i\beta^*-\bZ_i^{\top}\hat{\bgamma}=\varepsilon_i-U_i\beta^*+f(\bZ_i)-\bZ_i^{\top}\bgamma_0-\bZ_i^{\top}(\hat{\bgamma}-\bgamma_0);\\
&&W_i-\bZ_i^{\top}\hat{\btheta}=X_i-\bZ_i^{\top}\btheta_0+U_i+\bZ_i^{\top}(\btheta_0-\hat{\btheta}).
\end{eqnarray*}

As a result, we get:
\begin{eqnarray*}
T&=&\frac{1}{\sqrt n}\sum_{i=1}^n[(X_i-\bZ_i^{\top}\btheta_0+U_i)\{\varepsilon_i-U_i\beta^*+f(\bZ_i)-\bZ_i^{\top}\bgamma_0-\bZ_i^{\top}(\hat{\bgamma}-\bgamma_0)\}+\sigma^2_U\beta^*]\\
&&+\frac{(\btheta_0-\hat{\btheta})^\top}{\sqrt n}\sum_{i=1}^n\bZ_i(Y_i-W_i\beta^*-\bZ_i^{\top}\hat{\bgamma})\\
&=&\frac{1}{\sqrt n}\sum_{i=1}^n\{(X_i-\bZ_i^{\top}\btheta_0+U_i)(\varepsilon_i-U_i\beta^*+f(\bZ_i)-\bZ_i^{\top}\bgamma_0)+\sigma^2_u\beta^*\}\\
&&-\frac{1}{\sqrt n}\sum_{i=1}^n\eta_i \bZ_i^{\top}\times (\hat{\bgamma}-\bgamma_0)\\
&=&\frac{1}{\sqrt n}\sum_{i=1}^n\{(X_i-\bZ_i^{\top}\btheta_0+U_i)(\varepsilon_i-U_i\beta^*+f(\bZ_i)-\bZ_i^{\top}\bgamma_0)+\sigma^2_U\beta^*\}+o_p(1).
\end{eqnarray*}
Further note that
\begin{eqnarray*}
\hat\sigma^2&=&\frac{1}{n}\sum_{i=1}^n\left[\{X_i-\bZ_i^{\top}\btheta_0+U_i+\bZ_i^{\top}(\btheta-\hat{\btheta})\}
\{\varepsilon_i-U_i\beta^*+f(\bZ_i)-\bZ_i^{\top}\bgamma_0-\bZ_i^{\top}(\hat{\bgamma}-\bgamma_0)\}
+\sigma^2_U\beta^*\right]^2\\
&=&\frac{1}{n}\sum_{i=1}^n\{(X_i-\bZ_i^{\top}\btheta_0+U_i)(\varepsilon_i-U_i\beta^*+f(\bZ_i)-\bZ_i^{\top}\bgamma_0)+\sigma^2_U\beta^*\}^2+o_p(1).
\end{eqnarray*}
Thus $\hat\sigma^2$ is a consistent estimator of corresponding asymptotic variance of $T$. We then can conclude that $T_{DF}$ converges to standard normal distribution.

\subsection*{Proof of Theorem \ref{thm5.220220503}}
 Firstly, assume Model (a) holds. Denote $\Delta\beta=\beta_0-\beta^*=n^{-1/2}c$.
In this situation, we have
\begin{eqnarray*}
&&Y_i-W_i\beta^*-\bZ_i^{\top}\hat{\bgamma}=\varepsilon_i-U_i\beta^*+X_i\Delta\beta+\bZ_i^{\top}(\bgamma_0-\hat{\bgamma}).
\end{eqnarray*}
As a result, we get:
\begin{eqnarray*}
T&=&\frac{1}{\sqrt n}\sum_{i=1}^n[(\varepsilon_i-U_i\beta^*)\{X_i-\bZ_i^{\top}\btheta_0+U_i-\bZ_i^{\top}(\hat{\btheta}-\btheta_0)\}+\sigma^2_U\beta^*]\\
&&+\frac{(\bgamma_0-\hat{\bgamma})^\top}{\sqrt n}\sum_{i=1}^n\bZ_i(W_i-\bZ_i^{\top}\hat{\btheta})+\frac{c}{n}\sum_{i=1}^nX_i(W_i-\bZ_i^{\top}\hat{\btheta})\\
&=&\frac{1}{\sqrt n}\sum_{i=1}^n\{(\varepsilon_i-U_i\beta^*)(X_i-\bZ_i^{\top}\btheta_0+U_i)+\sigma^2_U\beta^*\}\\
&&+c\sigma^2_{X,Z}+o_p(1).
\end{eqnarray*}
Further note that
\begin{eqnarray*}
\hat\sigma^2&=&\frac{1}{n}\sum_{i=1}^n\left[\{\varepsilon_i-U_i\beta^*+X_i\Delta\beta+\bZ_i^{\top}(\bgamma_0-\hat{\bgamma})\}\{X_i-\bZ_i^{\top}\btheta_0+U_i-\bZ_i^{\top}(\hat{\btheta}-\btheta_0)\}
+\sigma^2_U\beta^*\right]^2\\
&=&\frac{1}{n}\sum_{i=1}^n\{(\varepsilon_i-U_i\beta^*)(X_i-\bZ_i^{\top}\btheta_0+U_i)+\sigma^2_U\beta^*\}^2+o_p(1).
\end{eqnarray*}

As a result, we obtain that:
\begin{eqnarray*}
T_{DF}\rightarrow N\left(c\frac{\sigma^2_{X,Z}}{\sigma_1},1\right).
\end{eqnarray*}

Alternatively, assume now that Model (b) holds.
In this situation, we have
\begin{eqnarray*}
&&Y_i-W_i\beta^*-\bZ_i^{\top}\hat{\bgamma}=\varepsilon_i-U_i\beta^*+X_i\Delta\beta+f(\bZ_i)-\bZ_i^{\top}\bgamma_0-\bZ_i^{\top}(\hat{\bgamma}-\bgamma_0).
\end{eqnarray*}

As a result, we get:
\begin{eqnarray*}
T&=&\frac{1}{\sqrt n}\sum_{i=1}^n[(X_i-\bZ_i^{\top}\btheta_0+U_i)\{\varepsilon_i-U_i\beta^*+f(\bZ_i)-\bZ_i^{\top}\bgamma_0-\bZ_i^{\top}(\hat{\bgamma}-\bgamma_0)\}+\sigma^2_U\beta^*]\\
&&+\frac{(\btheta_0-\hat{\btheta})^\top}{\sqrt n}\sum_{i=1}^n\bZ_i(Y_i-W_i\beta^*-\bZ_i^{\top}\hat{\bgamma})+\frac{c}{n}\sum_{i=1}^nX_i(W_i-\bZ_i^{\top}\hat{\btheta})\\
&=&\frac{1}{\sqrt n}\sum_{i=1}^n\{(X_i-\bZ_i^{\top}\btheta_0+U_i)(\varepsilon_i-U_i\beta^*+f(\bZ_i)-\bZ_i^{\top}\bgamma_0)+\sigma^2_U\beta^*\}\\
&&+c\sigma^2_{X,Z}+o_p(1).
\end{eqnarray*}
Further note that
\begin{eqnarray*}
\hat\sigma^2&=&\frac{1}{n}\sum_{i=1}^n\left[\{X_i-\bZ_i^{\top}\btheta_0+U_i+\bZ_i^{\top}(\btheta-\hat{\btheta})\}\right.\\&&\left.
\{\varepsilon_i-U_i\beta^*+X_i\Delta\beta+f(\bZ_i)-\bZ_i^{\top}\bgamma_0-\bZ_i^{\top}(\hat{\bgamma}-\bgamma_0)\}
+\sigma^2_U\beta^*\right]^2\\
&=&\frac{1}{n}\sum_{i=1}^n\{(X_i-\bZ_i^{\top}\btheta_0+U_i)(\varepsilon_i-U_i\beta^*+f(\bZ_i)-\bZ_i^{\top}\bgamma_0)+\sigma^2_U\beta^*\}^2+o_p(1).
\end{eqnarray*}
Thus, we can conclude that:
\begin{eqnarray*}
T_{DF}\rightarrow N\left(c\frac{\sigma^2_{X,Z}}{\sigma_2},1\right).
\end{eqnarray*}

\subsection*{Proof of Theorem \ref{thm5.420220502}}
Before proving Theorem \ref{thm5.420220502}, We state the following lemmas.
\begin{lemma}\label{lemma5.120220502}
Let $X_{1}, \ldots, X_{n}$ be independent mean-$0$ sub-exponential random variables. Then there exists constants $C, C^\prime >0$ such that for any $n$ and $t>0$, we have
$$\operatorname{P}\left(\frac{1}{n}\left|\sum_{i=1}^{n} X_{i}\right| \geqslant t\right) \leqslant C^\prime \exp \left\{-Ct^2 n\right\}.$$
\end{lemma}
\begin{proof}
The exponential type tail property can follow from Proposition 5.16 of \citet{vershynin2010introduction}.
\end{proof}
\begin{lemma}\label{lemma220220501}
$||\hat{\bgamma}-\bgamma_0||_1=o_p(\lambda_{\bgamma} s_{\bgamma}),||\hat{\btheta}-\btheta_0||_1=o_p(\lambda_{\btheta} s_{\btheta}).$
\end{lemma}
\begin{proof}
This lemma can be implied by Theorem 1 of \citet{loh2015regularized}.  \citet{loh2015regularized} consider the case for a general loss function $\mathcal{L}_n$. We can select the $\mathcal{L}_n$  be $\mathcal{L}_n(\beta,\bgamma)=\dfrac{1}{2n}\sum_{i=1}^n(Y_i-W_i\beta -\bZ_i^\top\bgamma)^2$ in (\ref{5.1220220501}) and $\mathcal{L}_n(\beta,\bgamma)=\dfrac{1}{2n}\sum_{i=1}^n(W_i- \bZ_i^\top\btheta)^2$ in (\ref{5.1320220501}).
\end{proof}
\begin{lemma}\label{lemma320220501}
$\dfrac{1}{n}\sum_{i=1}^n(\bZ_i^\top\tilde\bgamma)^2=O_p(\lambda_{\bgamma}^2s_{\bgamma})$, $\dfrac{1}{n}\sum_{i=1}^n(\bZ_i^\top\tilde\btheta)^2=O_p(\lambda_{\btheta}^2s_{\btheta})$.
\end{lemma}
\begin{proof}
This lemma can be implied by Theorem 2 of \citet{loh2015regularized} by similar way of proving Lemma \ref{lemma220220501}.
\end{proof}

\noindent{\bf{Proof of Theorem \ref{thm5.420220502}:}}
Firstly, assume that $Y$- model in (\ref{5.120220501}) holds while $X$- model in (\ref{5.120220501}) is misspecified. In this situation, we still have
\begin{eqnarray*}
&&Y_i-W_i\beta^*-\bZ_i^{\top}\hat{\bgamma}=\varepsilon_i-U_i\beta^*+\bZ_i^{\top}(\bgamma_0-\hat{\bgamma}).
\end{eqnarray*}
As a result, we get:
\begin{eqnarray*}
T&=&\frac{1}{\sqrt n}\sum_{i=1}^n\{(\varepsilon_i-U_i\beta^*)(X_i-\bZ_i^{\top}\hat{\btheta}+U_i)+\sigma^2_U\beta^*\}+\frac{(\bgamma_0-\hat{\bgamma})^\top}{\sqrt n}\sum_{i=1}^n\bZ_i(W_i-\bZ_i^{\top}\hat{\btheta})\\
&=&T_1+T_2.
\end{eqnarray*}
Different from the low dimensional situation, now $T_2$ is not exactly equal to zero. However, we can still control the term.
Actually, by the KKT condition for folded concave penalty (see Theorem 3.1 in \citet{fan2020statistical}), we have :
\begin{eqnarray*}
\|\frac{1}{n}\sum_{i=1}^n\bZ_i(W_i-\bZ_i^{\top}\hat{\btheta})\|_{\infty}\leq C\lambda_{\btheta}.
\end{eqnarray*}
Applying Lemma \ref{lemma220220501}, we can have:
\begin{eqnarray*}
T_2\leq \sqrt n\|\hat{\bgamma}-\bgamma_0\|_{1}\|\frac{1}{n}\sum_{i=1}^n\bZ_i(W_i-\bZ_i^{\top}\hat{\btheta})\|_{\infty}=O_p(\sqrt n\lambda_{\bgamma}\lambda_{\btheta}s_{\bgamma})=o_p(1).
\end{eqnarray*}

Next  turn to consider the term $T_1$. Different from low dimensional setting, we do not have an explicit form for $\hat{\btheta}$ now.
To deal with this problem, we may introduce $\btheta_0$ which is the limit of $\hat{\btheta}$. Then
\begin{eqnarray*}
T_1&=&\frac{1}{\sqrt n}\sum_{i=1}^n\{(\varepsilon_i-U_i\beta^*)(X_i-\bZ_i^{\top}\btheta_0+U_i)+\sigma^2_U\beta^*\}
+\frac{(\btheta_0-\hat{\btheta})^\top}{\sqrt n}\sum_{i=1}^n(\varepsilon_i-U_i\beta^*)\bZ_i\\
&=:&T_{11}+T_{12}.
\end{eqnarray*}
Clearly $T_{11}$ converges to normal distribution. We focus on $T_{12}$. Since $\varepsilon_i$, $U_i$ are $Z_{ij} $ are sub-Gaussian with uniform norm $K$. We have $(\varepsilon_i-U_i\beta^*)Z_{ij}$ to be sub-exponential with  $||\varepsilon_i-U_i\beta^*)Z_{ij}||_{\psi_1}\leq K_1<\infty$ holds uniformly.
By Lemma \ref{lemma5.120220502}, we  have
$$P\left(\frac{1}{n}\left\|\sum_{i=1}^n(\varepsilon_i-U_i\beta^*)Z_{i}\right\|_{\infty}>t\right) \leq  pC^\prime \exp \left(-C n t^2\right)$$
for constants $C, C^\prime >0$. Thus,
\begin{eqnarray*}
\|\frac{1}{\sqrt n}\sum_{i=1}^n(\varepsilon_i-U_i\beta^*)\bZ_i\|_{\infty}=O_p(\sqrt{\log p}),
\end{eqnarray*}
and hence 
$$T_{12}\leq||\btheta_0-\hat{\btheta}||_1\|\frac{1}{\sqrt n}\sum_{i=1}^n(\varepsilon_i-U_i\beta^*)\bZ_i\|_{\infty}=o_p(1).$$
Therefore,
\begin{equation}\label{eqnT}
T=\frac{1}{\sqrt n}\sum_{i=1}^n[(\varepsilon_i-U_i\beta^*)(X_i-\bZ_i^{\top}\btheta_0+U_i)+\sigma^2_U\beta^*]+o_p(1)
\end{equation}

\vspace{5mm}

Further we will show  consistency of $\hat{\sigma}$. Notice that
\begin{eqnarray}\label{eqnsigma}
\hat\sigma^2-\sigma_1^2=:I_0-D_1-D_2-D_3-D_4+D_5+D_6
\end{eqnarray}
$I_0$ and $D_1,\cdots D_6$ are given by
\begin{eqnarray*}
I_0&=&\dfrac{1}{n}\sum_{i=1}^n\{(\varepsilon_i-U_i\beta^*)(X_i-\bZ_i^{\top}\btheta_0+U_i)+\sigma^2_U\beta^*\}^2\\
&-&\mbE\{(\varepsilon_i-U_i\beta^*)(X_i-\bZ_i^{\top}\btheta_0+U_i)+\sigma^2_U\beta^*\}^2,\\
D_1&=&\frac{2}{n}\sum_{i=1}^n\{(\varepsilon_i-U_i\beta^*)(X_i-\bZ_i^{\top}\btheta_0+U_i)^2\times \bZ_i^\top\tilde\bgamma\},\\
D_2&=&\frac{2}{n}\sum_{i=1}^n\{\sigma_U^2\beta^*\times(X_i-\bZ_i^{\top}\btheta_0+U_i)\times \bZ_i^\top\tilde\bgamma\}\\
D_3&=&\frac{2}{n}\sum_{i=1}^n\{(\varepsilon_i-U_i\beta^*)^2(X_i-\bZ_i^{\top}\btheta_0+U_i)\times \bZ_i^\top\tilde\btheta\},\\
D_4&=&\frac{2}{n}\sum_{i=1}^n\{\sigma_U^2\beta^*\times(\varepsilon_i-U_i\beta^*)\times \bZ_i^\top\tilde\btheta\},\\
D_5&=&\frac{2}{n}\sum_{i=1}^n\{(\varepsilon_i-U_i\beta^*)(X_i-\bZ_i^{\top}\btheta_0+U_i)\times \bZ_i^\top\tilde\bgamma\times \bZ_i^\top\tilde\btheta\},\\
D_6&=&\frac{2}{n}\sum_{i=1}^n\{\sigma_U^2\beta^*\times \bZ_i^\top\tilde\bgamma\times \bZ_i^\top\tilde\btheta\},
\end{eqnarray*}
where $\tilde\btheta=\hat{\btheta}-\btheta_0$ and  $\tilde\bgamma=\hat{\bgamma}-\bgamma_0$.
We are going to show $D_i$ are all $o_p(1)$. By symmetry, we only need to deal with $D_1$ and $D_5$.

\begin{eqnarray*}
D_1&\leq&2\max_{1\leq i\leq n}\left|(\varepsilon_i-U_i\beta^*)(X_i-\bZ_i^{\top}\btheta_0+U_i)^2\right|\max_{1\leq i\leq n, 1\leq j\leq p}|Z_{i,j}|\|\tilde{\bgamma}\|_1\\
&=&O_P(\log n^{3/2}\log n\log p\lambda_{\bgamma} s_{\bgamma})=o_P(1).\\
\end{eqnarray*}
Under condition \ref{subgau}, we have $D_1=O_p(\lambda_{\bgamma}^2s_{\bgamma})$.
By the property of sub-Gaussian variables and Cauchy-Schwartz inequality,
\begin{eqnarray*}
D_5&\leq&2\max_{1\leq i\leq n}\left|(\varepsilon_i-U_i\beta^*)(X_i-\bZ_i^{\top}\btheta_0+U_i)\right|\frac{1}{n}\sum_{i=1}^{n}\left|\bZ_i^\top\tilde\bgamma\times    \bZ_i^\top\tilde\btheta\right|\\
&\leq& 2\max_{1\leq i\leq n}\left|(\varepsilon_i-U_i\beta^*)(X_i-\bZ_i^{\top}\btheta_0+U_i)\right|\sqrt{\frac{1}{n}\sum_{i=1}^{n}(\bZ_i^\top\tilde\bgamma)^2}\sqrt{\frac{1}{n}\sum_{i=1}^{n}(\bZ_i^\top\tilde\btheta)^2}\\
&\stackrel{(i)}{=}&O_p(\log(n)\lambda_{\bgamma}\lambda_{\btheta} \sqrt{s_{\bgamma} s_{\btheta}}).
\end{eqnarray*} 
(i) follows from Lemma \ref{lemma320220501}.
Next we consider the situation that $X$- model in (\ref{5.120220501}) holds while $Y$- model in (\ref{5.120220501}) is misspecified. In this situation, we have:
$$Y_i-W_i\beta^*-\bZ_i^\top\hat{\bgamma}=\varepsilon_i-U_i\beta^*+f(\bZ_i)-\bZ_i^\top\hat{\bgamma}.$$
\begin{eqnarray*}
T&=&\frac{1}{\sqrt n}\sum_{i=1}^n\{(W_i-\bZ_i^{\top}\btheta_0)(Y_i-W_i\beta^*-\bZ_i^\top\hat{\bgamma})+\sigma^2_U\beta^*\}+\frac{(\btheta_0-\hat{\btheta})^\top}{\sqrt n}\sum_{i=1}^n\bZ_i(Y_i-W_i\beta^*-\bZ_i^\top\hat{\bgamma}).
\end{eqnarray*}
Similarly, we can show that the second term is negligible, and thus we get:
\begin{eqnarray*}
T&=&\frac{1}{\sqrt n}\sum_{i=1}^n\{(W_i-\bZ_i^{\top}\btheta_0)(Y_i-W_i\beta^*-\bZ_i^\top\hat{\bgamma})+\sigma^2_U\beta^*\}+o_p(1)\\
&=&\frac{1}{\sqrt n}\sum_{i=1}^n\{(\eta_i+U_i)(\varepsilon_i-U_i\beta^*+f(\bZ_i)-\bZ_i^\top\bgamma_0)+\sigma^2_U\beta^*\}-\frac{1}{\sqrt n}\sum_{i=1}^n\{(\eta_i+U_i)(\bZ_i^\top\tilde\bgamma)\}+o_p(1)
\end{eqnarray*}
Meanwhile, we can similarly show $\hat\sigma^2$  is a consistent estimator of $\sigma^2_2$.

\subsection*{Proof of Theorem \ref{thm5.4420220502}}

Assume that $Y$- model in (\ref{5.120220501}) holds while model $X$-model in (\ref{5.120220501}) is misspecified.
To look at local power of the statistics, denote $\Delta\beta=\beta_0-\beta^*=n^{-1/2}c$.
In this situation, we have
\begin{eqnarray*}
&&Y_i-W_i\beta^*-\bZ_i^{\top}\hat{\bgamma}=\varepsilon_i-U_i\beta^*+X_i\Delta\beta+\bZ_i^{\top}(\bgamma_0-\hat{\bgamma}).
\end{eqnarray*}
\begin{eqnarray*}
\hat\sigma^2&=&\frac{1}{n}\sum_{i=1}^n\left\{[\varepsilon_i-U_i\beta^*+X_i\Delta\beta+\bZ_i^{\top}\tilde{\bgamma})][X_i-\bZ_i^{\top}\btheta_0+U_i-\bZ_i^{\top}\tilde{\btheta})]
+\sigma^2_U\beta^*\right\}^2
\end{eqnarray*}
Similar to (\ref{eqnsigma}), we can show
\begin{eqnarray*}
\hat\sigma^2-\sigma_1^2&=& o_p(1)+\frac{2}{n}\sum_{i=1}^n\{(\varepsilon_i-U_i\beta^*)(X_i-\bZ_i^{\top}\btheta_0+U_i)^2 X_i\Delta\beta\}\\
&+&\frac{2}{n}\sum_{i=1}^n\{\sigma_U^2\beta^*(X_i-\bZ_i^{\top}\btheta_0+U_i) X_i\Delta\beta\}\\
&=:&o_p(1)+E_1+E_2
\end{eqnarray*}
Under condition \ref{subgau},
\begin{equation}
E_1\leq \max_{1\leq i\leq n}\left|(\varepsilon_i-U_i\beta^*)(X_i-\bZ_i^{\top}\btheta_0+U_i)^2X_i\right|\Delta\beta=O_P({n^{-\frac{1}{2}}(\log n)^2})=o_p(1).
\end{equation}
Similarly,
\begin{equation}
E_2=O_P({n^{-\frac{1}{2}}\log n})=o_p(1).
\end{equation}
Therefore $\hat\sigma$ is still a consistent estimator of $\sigma^2$.
\begin{eqnarray*}
T&=&\frac{1}{\sqrt n}\sum_{i=1}^n[(\varepsilon_i-U_i\beta^*)(X_i-\bZ_i^{\top}\hat{\btheta}+U_i)+\sigma^2_U\beta^*]+\frac{1}{\sqrt n}\sum_{i=1}^n X_i(X_i-\bZ_i^\top\hat{\btheta}+U_i)\Delta\beta\\
&&+\frac{(\bgamma_0-\hat{\bgamma})^\top}{\sqrt n}\sum_{i=1}^n\bZ_i(W_i-\bZ_i^{\top}\hat{\btheta}),
\end{eqnarray*}
It follows from (\ref{eqnT}) that
\begin{eqnarray*}
T&=&\frac{1}{\sqrt n}\sum_{i=1}^n[(\varepsilon_i-U_i\beta^*)(X_i-\bZ_i^{\top}\btheta_0+U_i)+\sigma^2_U\beta^*]+\frac{1}{\sqrt n}\sum_{i=1}^n X_i(X_i-\bZ_i^\top\hat{\btheta}+U_i)\Delta\beta+o_p(1).
\end{eqnarray*}
\begin{eqnarray*}
\frac{1}{\sqrt n}\sum_{i=1}^n X_i(X_i-\bZ_i^\top\hat{\btheta}+U_i)\Delta\beta=\frac{1}{\sqrt n}\sum_{i=1}^n X_i(X_i-\bZ_i^\top\btheta_0+U_i)\Delta\beta-\frac{1}{\sqrt n}\sum_{i=1}^n X_i\bZ_i^\top\tilde\btheta\Delta\beta=:T_3-T_4
\end{eqnarray*}
Let $\sigma^2_{X,Z}:=\bSigma_{X,X}-\bSigma^\top_{X,Z}\bSigma^{-1}_{Z,Z}\bSigma_{X,Z}.$
Under condition \ref{consigma}, we have 
\begin{eqnarray*}
\sigma^2_{X,Z}=\left(1,-\boldsymbol{\omega}^\top_{X,Z}\right)\bSigma\left(\begin{array}{c}
1 \\
-\boldsymbol{\omega}_{X,Z}
\end{array}\right)\geq \underline{C}.
\end{eqnarray*}
By the law of large numbers, it follows that
\begin{eqnarray*}
T_3=\frac{c}{n}\sum_{i=1}^n X_i(X_i-\bZ_i^\top\btheta_0+U_i)=c\sigma^2_{X,Z}+o_p(1)
\end{eqnarray*}
\begin{eqnarray*}
T_4=O_P(\sqrt{\frac{\log p}{n}})=o_p(1).
\end{eqnarray*}
Therefore $T_{DF}=N\left(\dfrac{c\sigma^2_{X,Z}}{\sigma_2},1\right)+o_p(1).$

Similar result holds under the Model (b).

\subsection*{Proof of Theorem \ref{them5.520220501}}

In the proof we assume Model (c) holds. The proof under the case that Model (d) holds is the same due to symmetry. The following lemmas are required for the proof of the main theorem.

\begin{lemma}\label{lemma5.420220502}
Suppose Model (c) holds and that $\beta_0=\beta^*+hn^{-1/2}$, then with probability tending to one, $\bgamma_0$ lies in the feasible set of the optimization problem (\ref{5.1520220502})-(\ref{5.2020220430}).
\end{lemma}
\begin{proof}
It suffices for us to show that it holds with probability tending to $1$ that
\begin{eqnarray}
\left\|n^{-1} \mathbf{Z}^{\top}(\mathbf{V}-\mathbf{Z} {\bgamma_0})\right\|_{\infty} &\leq& \eta_{\bgamma}  \label{5.2920220502}\\ \|\mathbf{V}-\mathbf{Z} {\bgamma_0}\|_{\infty}& \leq& \mu_{\bgamma} \label{5.3020220502}\\
n^{-1} \mathbf{V}^{\top}(\mathbf{V}-\mathbf{Z} {\bgamma_0}) &\geq& \rho_\eta n^{-1}\|\mathbf{V}\|_{2}^{2}.\label{5.3120220502}
\end{eqnarray}
For (\ref{5.2920220502}), notice that $\bV-\bZ\bgamma_0=hn^{-1/2}\bX-U\beta^*+\bepsilon$. By Lemma \ref{lemma5.120220502},  we  have
\begin{eqnarray*}
&&P\left(\frac{1}{n}\left\|\sum_{i=1}^n\left[(hn^{-1/2}X_i+\varepsilon_i-U_i\beta^*)Z_{i}-\mbE\{(hn^{-1/2}X_i+\varepsilon_i-U_i\beta^*)Z_{i}\}\right]\right\|_{\infty}>t\right)\\
&\leq& pC^\prime \exp \left(-C n t^2\right).
\end{eqnarray*}
for constants $C, C^\prime>0$. We have 
\begin{eqnarray*}
&&|\mbE\{(hn^{-1/2}X_i+\varepsilon_i-U_i\beta^*)Z_{i,j}\}|= |\mbE(hn^{-1/2}Z_{i,j}\bZ_i^\top\btheta_0)|\\
&=&hn^{-1/2}(\bSigma_{Z,Z})_{j:}\btheta_0 \leq hn^{-1/2}\lambda_{max}(\bSigma)\|\btheta_0\|_2\leq G_1 n^{-1/2},
\end{eqnarray*}
for constant $G_1$.
Therefore,
\begin{eqnarray*}
&&P(\frac{1}{n}\|\bZ^\top(\bV-\bZ\bgamma_0)\|_\infty> \eta_{\bgamma})\\
&\leq& P(\frac{1}{n}\|\bZ^\top(\bV-\bZ\bgamma_0)-\mbE\{\bZ^\top(\bV-\bZ\bgamma_0)\}\|_\infty>\eta_{\bgamma}-G_1n^{-1/2})\\
&\leq& pC^\prime\exp\{C(\eta_{\bgamma}-G_1n^{1/2})\}.
\end{eqnarray*}
For a constant $C_1$ such that $\eta_{\bgamma}\geq C_1\log n\sqrt{n^{-1}\log p}$, we have  (\ref{5.2920220502}) holds with probability tending to 1.

To tackle (\ref{5.3020220502}), we would like to show that there exists a constant $C_2>0$ such that for any $\mu_{\bgamma}>C_2\sqrt{\log n}$, $\|\bV-\bZ\bgamma_0\|_\infty\leq\mu_{\bgamma}$ with probability tending to 1. Noticing that each element of  $\bV-\bZ\bgamma_0=hn^{-1/2}\bX-U\beta^*+\bepsilon$ is a sub-Gaussian random variable, therefore there exists a constant $G_2>0$ such that for any $C_2>0$, $P(\|\bV-\bZ\bgamma_0\|_\infty>C_2\sqrt{\log n})\leq n\exp(1-G_2C_2^2\log n).$ The results follows by taking $C_2$ large enough.

For (\ref{5.3120220502}), we are going to show that for any $c_0>0$ and $\rho_{\bgamma}\in(0,\sigma_\varepsilon^2+\sigma_U^2\beta^{*2}-c_0)$, we have
\begin{eqnarray}
P(\frac{1}{n}\mathbf{V}^{\top}(\mathbf{V}-\mathbf{Z} {\bgamma_0})\leq\rho_{\bgamma})\rightarrow 0.
\end{eqnarray}
Noticing that
\begin{eqnarray}
&&\dfrac{1}{n}\mathbf{V}^{\top}(\mathbf{V}-\mathbf{Z} {\bgamma_0})=\dfrac{1}{n}\sum_{i=1}^nV_i(V_i-\bZ_i^\top\bgamma_0)\\
&=&\dfrac{1}{n}\sum_{i=1}^n(n^{-1/2}hX_i+\varepsilon_i-U_i\beta^*+\bZ_i^\top\bgamma_0)(n^{-1/2}hX_i+\varepsilon_i-U_i\beta^*)\\
&=&\frac{1}{n}\sum_{i=1}^n(n^{-1/2}hX_i+\varepsilon_i-U_i\beta^*)^2+\frac{1}{n}\sum_{i=1}^n(n^{-1/2}hX_i+\varepsilon_i-U_i\beta^*)\bZ_i^\top\bgamma_0.
\end{eqnarray}
By the law of large numbers,
$$\dfrac{1}{n}\sum_{i=1}^n(n^{-1/2}hX_i+\varepsilon_i-U_i\beta^*)^2=\mbE(\varepsilon_i-U_i\beta^*)^2+o_P(1).$$
On the other hand,
\begin{eqnarray*}
&&\mbE\{(n^{-1/2}hX_i+\varepsilon_i-U_i\beta^*)\bZ_i^\top\bgamma_0\}\\&=&\mbE\{(n^{-1/2}h\bZ_i^\top\btheta_0+\varepsilon_i-U_i\beta^*)\bZ_i^\top\bgamma_0\}\\
&\leq&n^{-1/2}h\|\bgamma_0\|_2\|\btheta_0\|_2\lambda_{max}(\bSigma)=o_p(1),
\end{eqnarray*}
and
\begin{eqnarray*}
&&Var((n^{-1/2}hX_i+\varepsilon_i-U_i\beta^*)\bZ_i^\top\bgamma_0\}\\&\leq&\mbE\{(n^{-1/2}hX_i+\varepsilon_i-U_i\beta^*)^2(\bZ_i^\top\bgamma_0)^2\}\\
&\leq& \sqrt{\mbE(n^{-1/2}hX_i+\varepsilon_i-U_i\beta^*)^2 \mbE(\bZ_i^\top\bgamma_0)^2}
\end{eqnarray*}
 is bounded above.
Therefore $\dfrac{1}{n}\mathbf{V}^{\top}(\mathbf{V}-\mathbf{Z} {\bgamma_0})=\mbE(\varepsilon_i-U_i\beta^*)^2+o_P(1)=\sigma_U^2\beta^{*2}+o_P(1)$.
This ends the proof of Lemma \ref{lemma5.420220502}.
\end{proof}

\begin{lemma}\label{lemma5.520220502}
Let $\bX\in \mathbb{R}^n$, $\bZ\in\mathbb{R}^{n\times p}$ and $C\subset\mathbb{R}^p$ . Suppose there exists $\btheta_0\in C$ such that $\|n^{-1}\bZ^\top(\bX-\bZ\btheta_0)\|_\infty \leq \eta$. Let $\hat{\btheta}=\arg\min_{\btheta\in C}\|\btheta\|_1 $subject to $\|n^{-1}\bZ^\top(\bX-\bZ\btheta)\|_\infty\leq\eta$. If $s_{\btheta}=\|\btheta_0\|_0$ and if there exists $\kappa$ such that the restricted eigenvalue condition holds, that is
$$\min _{J_{0} \subseteq\{1, \ldots, p\}\atop \left|J_{0}\right| \leq s}\min _{\delta \neq 0\atop\left\|{\delta}_{J_{0}^{c}}\right\|_{1} \leq\left\|{\delta}_{J_{0}}\right\|_{1}}\frac{\left\| \bZ {\delta}\right\|_{2}}{\sqrt{n}\left\|{\delta}_{J_{0}}\right\|_{2}} \geq \kappa.$$
Then for $\delta=\hat{\btheta}-\btheta_0$ it holds that $\|\delta\|_1\leq 8\eta s \kappa^{-2}$ and $n^{-1} {\delta}^{\top} \bZ^\top \bZ{\delta} \leq 16 \eta^{2} s \kappa^{-2}.$
\end{lemma}
\begin{proof}
It follows from Lemma 6 in \citet{bradic2020fixed}.
\end{proof}

\noindent{\bf{Proof of Theorem \ref{them5.520220501}.}}

By Lemma \ref{lemma5.420220502}, we have that $P(\mathcal{A})\rightarrow 1,$ where
$$\mathcal{A}=\{\btheta_0 \text{ and } \bgamma_0 \text { are feasible for }  (\ref{5.1520220502})-(\ref{5.2020220430}) \text{ and } (\ref{5.2120220502})-(\ref{5.2420220502})\}.$$
This together with the restricted eigenvalue condition and Lemma \ref{lemma5.520220502} implies that with probability tending to 1, $\|\hat{\btheta}-\btheta_0\|=O_P(\eta_{\btheta} s_{\btheta})$ and $n^{-1}\|\bZ(\hat{\btheta}-\btheta_0)\|_2^2=O_p(\eta_{\btheta}^2 s_{\btheta}).$

Let $\hat{\varepsilon}_i=V_i-\bZ_i^\top\hat{\bgamma}.$ $\hat{\zeta}_i=W_i-\bZ_i^\top\hat{\btheta}$. Then $T=n^{-1/2}\sum_{i=1}^n(\hat{\zeta}_i\hat\varepsilon_i+\sigma_U^2\beta^*)$. Define $\xi_i=n^{-1/2}\{(\eta_i+U_i)\hat\varepsilon_i+\sigma^2_U\beta^{*}\}$. We have 
\begin{eqnarray*}
T-\sum_{i=1}^n\xi_i&=&n^{-1/2}\sum_{i=1}^n \bZ_i^\top(\btheta_0-\hat{\btheta})\hat\varepsilon_i\\
&\leq&n^{-1/2}\|\bZ^\top(\bV-\bX\hat{\bgamma})\|_\infty\|\hat{\btheta}-\btheta_0\|_1=O_P(\sqrt{n}\eta_{\bgamma}\eta_{\btheta}\|\btheta\|_0)=o_P(1).
\end{eqnarray*}
For the varaince term, we have 
\begin{eqnarray*}
\hat\sigma^2-\sum_{i=1}^n\xi^2_i&=&\frac{2}{n}\sum_{i=1}^n\bZ_i^\top(\hat{\btheta}-\btheta_0)(\eta_i+U_i)\hat\varepsilon_i+\frac{2}{n}\sum_{i=1}^n\bZ_i^\top(\btheta_0-\hat{\btheta})\sigma_U^2\beta^*+\frac{1}{n}\|\bZ(\hat{\btheta}-\btheta_0)\|_2^2\\
&:=&D_{n1}+D_{n2}+D_{n3}.
\end{eqnarray*}
$D_{n1}\leq2n^{-1}\|\hat{\btheta}-\btheta_0\|_1\|\bZ^\top(\bV-\bZ^\top \hat{\bgamma})\|_\infty \max\limits_{1\leq i\leq n}(\eta_i+U_i)=O_P(\eta_{\bgamma} \eta_{\btheta}\|\btheta_0\|_0\log n)=o_P(1).$
$D_{n2}=O_P(\max\limits_{1\leq i\leq n, 1\leq j\leq p}Z_{i,j}\|\hat{\btheta}-\btheta_0\|_1)=O_p((\log p)(\log n) \eta_{\btheta} s_{\btheta})=o_p(1).$
$D_{n3}=O_P(\eta_{\bgamma}^2s_{\bgamma})=o_P(1)$ by Lemma \ref{lemma5.520220502}.

Due to condition \ref{C1320220502}, it now suffices for us to show that $\dfrac{\sum_{i=1}^n\xi_i}{\sqrt{\sum_{i=1}^n}\xi_i^2}\stackrel{d}\rightarrow N(0,1)$. Define $\mathcal{F}_{n,i}$ as the $\sigma$-algebra generated by $(\bX,\bepsilon,\eta_1,\cdots,\eta_i,U_1,\cdots,U_i)$. 
We have 
\begin{eqnarray*}
\mbE(\xi_{i+1}|\mathcal{F}_{n,i})&=&n^{-1/2}\mbE(\eta_{i+1}(\varepsilon_{i+1}-U_{i+1}\beta^*+Z_{i+1}^\top\bgamma_0-Z_{i+1}^\top\hat{\bgamma})|\mathcal{F}_{n,i})+\\
&&n^{-1/2}\mbE(U_{i+1}(\varepsilon_{i+1}-U_{i+1}\beta^*+Z_{i+1}^\top\bgamma_0-Z_{i+1}^\top\hat{\bgamma})|\mathcal{F}_{n,i})-\sigma_U^2\beta^*\\
&=& 0.
\end{eqnarray*}

We are going to apply the martingale central limit theorem. By Theorem 3.4 of \citet{hall2014martingale}, it suffices for us to verify the following conditions.
\begin{itemize}    
\item[(i)] $\mbE \max_{1\leq n}\xi_i^2=o(1)$.
\item[(ii)] $\max_{1\leq n}|\xi_i|=o_p(1)$.
\end{itemize}
Using the inequality $(a+b)^2 \leq 2a^2+2b^2$,
\begin{eqnarray*}
E \max_{1\leq i\leq n}\xi_i^2&\leq&2 \mbE( n^{-1} \max_{1\leq i\leq n}(\eta_i+U_i)^2\|\bV-\bZ\hat{\bgamma}\|_\infty) +2\sigma_U^2\beta^*/n\\
&\leq&2n^{-1}\mu_{\bgamma} \mbE(\max_{1\leq i\leq n}(\eta_i+U_i)^2)+o(1)\\
&=&O_P((\log^{3/2} n)/n)+o(1)=o_p(1).
\end{eqnarray*}
The last equality is due to the fact that $\eta_i+U_i$ is  sub-Gaussian and $\mu_{\bgamma}=O(\sqrt{\log n})$.
Therefore we show claim (i) holds. Claim (ii) follows by Claim (i).


\end{document}